\documentclass[11pt,letter]{article}
\usepackage[margin=1in]{geometry}

\title{Parameterized Streaming Algorithms for Vertex Cover}

\author{
Rajesh Chitnis\thanks{Department of Computer Science , University of Maryland at
College Park, USA.  rchitnis@cs.umd.edu. Supported in part by NSF CAREER award 1053605, NSF grant CCF-1161626, ONR YIP award
N000141110662, DARPA/AFOSR grant FA9550-12-1-0423 and a Simons Award for Graduate Students in Theoretical Computer Science.}
\and
Graham Cormode \thanks{Department of Computer Science, University of
  Warwick, UK. g.cormode@warwick.ac.uk.}
    \and
 MohammadTaghi Hajiaghayi\thanks{Department of Computer Science , University of Maryland,
 USA.  hajiagha@cs.umd.edu. Supported in part by NSF CAREER award 1053605, NSF grant CCF-1161626, ONR YIP award
N000141110662, and DARPA/AFOSR grant FA9550-12-1-0423.}
    \and
Morteza Monemizadeh\thanks{Goethe-Universit\"{a}t Frankfurt, Germany and Department of Computer Science, University of Maryland at
College Park, USA.   monemizadeh@em.uni-frankfurt.de,
Supported in part by MO 2200/1-1.}
}

\usepackage{amsmath,amssymb,amsfonts,amsthm}
\usepackage{latexsym}
\usepackage{subfigure}
\usepackage{graphicx, datetime}
\usepackage{times}
\usepackage[scaled=0.92]{helvet} 
\usepackage{euler}
\usepackage{fancybox}
\usepackage{hyperref}
\usepackage{ifpdf}
\ifpdf
\fi

\newcount\shortyear\newcount\shorthour\newcount\shortminute
\shorthour=\time\divide\shorthour by 60\shortyear=\shorthour
\multiply\shortyear by 60\shortminute=\time\advance\shortminute by
-\shortyear \shortyear=\year\advance\shortyear by -1900

\def\zeit{\number\shorthour:\ifnum\shortminute<10 0\number\shortminute
\else\number\shortminute\fi}

\newtheorem{theorem}{Theorem}
\newtheorem{corollary}[theorem]{Corollary}
\newtheorem{lemma}[theorem]{Lemma}

\newtheorem{definition}[theorem]{Definition}

\newtheorem*{rep@theorem}{\rep@title} \newcommand{\newreptheorem}[2]{%
\newenvironment{rep#1}[1]{%
\def\rep@title{\bf #2 \ref{##1}}%
\begin{rep@theorem}}%
{\end{rep@theorem} } }
\makeatother
\newreptheorem{theorem}{Theorem}
\newreptheorem{lemma}{Lemma}

\newenvironment{proofof}[1]{\begin{trivlist} \item {\bf Proof
#1:~~}}
  {\qed\end{trivlist}}

\renewcommand{\Pr}[1]{\ensuremath{\mathbf{Pr}[#1]}}

\renewcommand{\paragraph}[1]{\medskip \noindent {\bf #1}}

\newcommand{\COMMENTED}[1]{{}}
\newcommand{\junk}[1]{\COMMENTED{#1}}
\newcommand{\NAT}{\ensuremath{\mathbb{N}}}
\newcommand{\NATURAL}{\NAT}
\newcommand{\REAL}{\ensuremath{\mathbb{R}}}

\newlength{\savedparindent}

\begin{document}


\sloppy
\maketitle

\begin{abstract}
As graphs continue to grow in size, we seek ways to effectively process such data at scale.
The model of streaming graph processing, in which a compact summary is maintained as each
edge insertion/deletion is observed, is an attractive one.
However, few results are known for optimization problems over such dynamic graph streams.

In this paper, we introduce a new approach to handling graph streams,
by instead seeking solutions for the parameterized versions of these problems
where we are given a parameter $k$ and the objective is to decide
whether there is a solution bounded by $k$.
By combining kernelization techniques with randomized sketch structures,
we obtain the first streaming algorithms for the parameterized versions of
the Vertex Cover problem. We consider two models for a graph stream on
$n$ nodes:
the insertion-only model where the edges can only be added,
and the dynamic model where edges can be both inserted and deleted.
More formally, we show the following results:

\begin{itemize}
\item In the insertion only model, there is a one-pass deterministic algorithm
for the parameterized Vertex Cover problem which computes a sketch
using $\tilde{O}(k^{2})$ space\footnote{$\tilde{O}(f(k))=O(f(k)\cdot
  \log^{O(1)} m)$, where $m$ is the number of edges.} such that at each
  timestamp in time $\tilde{O}(2^{2k^2})$ it can either extract a solution of
  size at most $k$ for the current instance, or report that no such solution exists.

\item In the dynamic model, and under the \emph{promise} that at each
timestamp there is a solution of size at most $k$, there is a one-pass
algorithm for the parameterized Vertex Cover problem which computes
a sketch using $\tilde{O}(k^{2})$ space such that in time $\tilde{O}(2^{2k^2})$ it can
extract a solution for the final instance with probability $1-\delta/{n^{O(1)}}$, where $\delta<1$.
To the best of our knowledge, this is the first graph streaming algorithm that combines linear
sketching with sequential operations that depend on the graph at the current time.

\item In the dynamic model without any promise, there is a one-pass
randomized algorithm for the parameterized Vertex Cover problem
which computes a sketch using $\tilde{O}(nk)$ space such that in time $\tilde{O}(nk+2^{2k^2})$ it can either extract a solution of size
at most $k$ for the final instance, or report that no such solution exists.
\end{itemize}
%
%
We also show a tight lower bound of $\Omega(k^2)$ for the space
complexity of any (randomized) streaming algorithms for
the parameterized Vertex Cover, even in the insertion-only model.
\end{abstract}



\newpage

\section{Introduction}
Many large graphs are presented in the form of a sequence of edges.
This stream of edges may be a simple stream of edge arrivals, where
each edge adds to the graph seen so far, or may include a mixture of
arrivals and departures of edges.
In either case, we want to be able to quickly answer basic optimization
questions over the current state of the graph, such as finding a
(maximal) matching over the current graph edges, or finding a
(minimum) vertex cover, while storing only a limited amount of
information, sublinear in the size of the current graph.

The semi-streaming model introduced by
Feigenbaum, Kannan, McGregor, Suri and Zhang \cite{FKMSZ05}  is a classical
streaming model in which maximal matching and vertex cover are studied.
In the semi-streaming model we are interested to solve (mostly approximately) graph problems
using one pass over the graph and using $O(n\ \text{polylog}\ n)$ space.
Numerous problems have been studied in this setting, such as
maintaining random walks and page rank over large graphs~\cite{DGP08}.

However, in many real world applications, we often observe instances of graph problems
whose solutions are small comparing to the size of input. Consider for example
the problem of finding the minimum number of fire stations to cover an entire city,
or other cases where we expect a small number of facilities will serve
a large number of locations. In these scenarios, assuming that the number of fire
stations or facilities is a small number $k$ is very practical. So, it will be interesting
to solve instances of graph problems (like minimal matching
and vertex cover) whose solutions are small (say, sublinear in the input size)
in a streaming fashion using space which is bounded with respect to
the size of their solutions, not the input size.

In this paper to solve graph problems whose solutions are small,
we {\em parameterize} problems with a parameter $k$, and
solve the decision problem of finding whether there exists a solution whose size is bounded by $k$.
We therefore seek \emph{parameterized streaming algorithms} whose space and time complexities
are bounded with respect to $k$ , i.e.,  sublinear in the size of the input.

There are several ways to formalize this question, and we give results for the most natural formalizations.
The basic case is when the input consists of a sequence of edge
arrivals only, for which we seek a {\em parameterized streaming
  algorithm} (PSA).
%
%
More challenging problems arise when the input stream is more
dynamic, and can contain both deletions and insertions of edges. In this case we seek a {\em dynamic parameterized streaming
  algorithm} (DPSA).
The challenge here is that when an edge in the matching is deleted, we
sometimes need substantial work to repair the solution, and have to
ensure that the algorithm has enough information to do so, while keeping only a bounded amount of working space. If we are \emph{promised} that at every timestamp there is a solution of cost $k$, then we seek a {\em promised dynamic parameterized streaming
  algorithm} (PDPSA).



\subsection{Parameterized Complexity}

Most interesting optimization problems on graphs are NP-hard,
implying that, unless P=NP, there is no polynomial time algorithm
that solves all the instances of an NP-hard problem exactly.
However as noted by Garey and Johnson~\cite{garey-johnson},
hardness results such as NP-hardness should merely constitute
the beginning of research.
The traditional way to combat  intractability is to design approximation
algorithms or randomized algorithms which run in polynomial time.
These methods have their own shortcomings:
we either get an approximate solution or lose the guarantee that
the output is always correct.

Parameterized complexity is essentially a two-dimensional analogue of ``P
vs NP". The running time is analyzed in finer detail:
instead of expressing it as a function of only the input size $n$, one
or more parameters of the input instance are defined,
and we investigate the effects of these parameters on the running time.
The goal is to design algorithms that work efficiently if
the parameters of the input instance are small, even if the size of
the input is large. We refer the reader to~\cite{DF99, FG06} for more background.

A  \textit{parameterization} of a decision problem  $P$ is a function
that assigns an integer parameter $k$ to each instance $I$ of $P$.
We assume that instance $I$ of problem $P$ has
the corresponding input $X=\{x_1,\cdots,x_i,\cdots, x_m\}$
consisting of elements $x_i$ (e.g. edges defining a graph).
We denote the input size of instance  $I$ by $|I|=m$.
In what follows, we assume that $f(k)$ and $g(k)$ are functions of an
integer parameter $k$.

\begin{definition}[Fixed-Parameter Tractability (FPT) ]
\label{def:fpt}
A parameterized problem $P$ is \emph{fixed-parameter tractable} (FPT) if there is an algorithm
that in time $f(k)\cdot m^{O(1)}$ returns a solution for each instance $I$ whose  size
fulfills a given condition corresponding to $k$ (say, at most  $k$ or at least $k$)
or reports that such a solution does not exist.
\end{definition}


To illustrate this concept,
we define the parameterized version of Vertex Cover as follows.
A \textit{vertex cover} of an undirected graph $G=(V,E)$ is a subset $S$
of vertices such that for every edge $e\in E$ at least one
of the endpoints (or vertices) of $e$ is in $S$.

\begin{definition}[Parameterized Vertex Cover ($VC(k)$)]
\label{def:vc}
Given an instance $(I,k)$ where $I$ is an undirected graph $G=(V,E)$ (with input size
$|I|=|E|=m$ and $|V|=n$) and parameter $k\in \NATURAL$,  the goal in the
\textit{parameterized Vertex Cover} problem ($VC(k)$  for short)
is to develop an algorithm that in time $f(k)\cdot m^{O(1)}$ either returns
a vertex cover of size at most $k$ for $G$, or reports that $G$ does not have
any vertex cover of size at most $k$.
\end{definition}

A simple branching method gives a $2^{k}\cdot m^{O(1)}$ algorithm for $VC(k)$:
choose any edge and branch on choosing either end-point of the edge into the solution.
The current fastest FPT algorithm for $VC(k)$ is due to
Chen et al.~\cite{kanj-vc} and runs in time $1.2738^{k}\cdot m^{O(1)}$.


One of the techniques used to obtain FPT algorithms is
\textit{kernelization}. In fact, it is known that a problem is FPT if
and only if it has a kernel~\cite{FG06}.
Kernelization has been used to design efficient algorithms by using polynomial-time
preprocessing to replace the input by another equivalent input of smaller size. More formally, we have:

\begin{definition}[Kernelization]
\label{def:kernel}
For a parameterized problem $P$, kernelization is a polynomial-time transformation
that maps an instance $(I,k)$ of $P$ to an instance $(I',k')$ such that
\begin{itemize}
\item $(I,k)$ is a yes-instance if and only if $(I',k')$ is a yes-instance;
\item $k'\le g(k)$ for some computable function $g$;
\item the size of $I'$  is bounded by some computable function $f$ of $k$, i.e., $|I'|\le f(k)$.
\end{itemize}
The output $(I',k')$ of a kernelization algorithm is called a \textit{kernel}.
\end{definition}

In Section~\ref{sec:kernel:VC} we review
the kernelization algorithm of Buss and Goldsmith~\cite{BG93} for the parameterized Vertex Cover problem
which relies on finding a maximal matching of a graph $G=(V,E)$.
This kernel gives a graph with $O(k^2)$ vertices and $O(k^2)$ edges.
Another kernelization algorithm given in~\cite{FG06} exploits the half-integrality
property of LP-relaxation for vertex cover due to Nemhauser and Trotter,
and produces a graph with at most $2k$ vertices.


\subsection{Parameterized Streaming Algorithms: Our Results}

In order to state our results for \textit{parameterized streaming}
we first define the notion of a \textit{sketch}  in a very general form.

\begin{definition}[Sketch \cite{AMS99,FKSV02,I06JACM}]
\label{def:sketch}
A \it{sketch} is a sublinear-space data structure that supports
a fixed set of queries and updates.
\end{definition}


%
\paragraph{Insertion-Only Streaming.}
Let  $P$ be a problem parameterized by $k\in \NATURAL$.
Let $I$ be an instance of $P$ that has the input
$X=\{x_1,\cdots,x_i,\cdots, x_m\}$.
Let $S$ be a stream of $\textsc{Insert}(x_i)$ (i.e., the insertion of an element $x_i$)
operations of underlying instance $(I,k)$.
In particular, stream $S$ is a permutation
$X'=\{x'_1,\cdots,x'_i,\cdots, x'_m\}$ for $x'_i\in X$ of an input $X$. Here we denote
the time when an input $x'_i$ is inserted by \textit{time} $i$. At time $i$,
the input which corresponds to instance $I$ is $X'_i=\{x'_1,\cdots,x'_i\}$.

\begin{definition}({\sc Parameterized streaming algorithm (PSA)})
Given stream $S$, let  $\mathcal{A}$ be an algorithm  that computes a sketch for problem $P$
using $\tilde{O}(f(k))$-space
and with one pass over stream $S$. Suppose at a time $i$,
algorithm $\mathcal{A}$ in time $\tilde{O}(g(k))$ extracts,  from the sketch,
a solution for input $X'_i$ (of instance $I$) whose  size fulfills
the condition corresponding to $k$  or reports that such a solution does not exist.
Then we say $\mathcal{A}$ is a $(f(k), g(k))$-PSA.
\label{defn:psa}
\end{definition}


%
For many problems, whether or not there is a solution of size at most
$k$ is monotonic under edge additions, and so
if at time $i$, algorithm $\mathcal{A}$ reports that a solution
for input $X'_i$ does not exist, then
there is also no solution for any input $X'_t$ of instance $I$
at all times $t>i$.
Consequently, we can terminate the algorithm $\mathcal{A}$.
We state our result on the parameterized streaming
algorithm for Vertex Cover and prove it in Section \ref{sec:psa:vc}.

\begin{theorem}
\label{THM:VC:INSERTION}
Let $S$ be a stream of  insertions of edges of an underlying graph $G$.
Then there exists a \emph{deterministic} $(k^2,2^{2k^2})$-PSA for $VC(k)$ problem.
\end{theorem}


The best known kernel size for the $VC(k)$ problem is $O(k^2)$ edges~\cite{BG93}.
In fact, Dell and van Melkebeek~\cite{dell-lower-bound} showed that it is not possible to get
a kernel for the $VC(k)$ problem with $O(k^{2-\epsilon})$ edges for any $\epsilon>0$,
under some assumptions from classical complexity.
Interestingly, the space complexity of our PSA of Theorem \ref{THM:VC:INSERTION} matches
this best known kernel size.
In Section \ref{sec:vc:bound} we show that the space complexity of above PSA
is optimal even if we use randomization.
More precisely, we prove the following result.

\begin{theorem}
\label{thm:vc:lower:bound}
Any (randomized) PSA for the $VC(k)$ problem requires $\Omega(k^2)$ space.
\label{thm:lower-bound-vc}
\end{theorem}




\paragraph{Dynamic Streaming.}
We define \textit{dynamic parameterized stream} as a generalization of \textit{dynamic graph stream}
introduced by Ahn, Guha and McGregor \cite{AGM12a}.

\begin{definition}[Dynamic Parameterized Stream]
Let  $P$ be a problem parameterized by $k\in \NATURAL$. Let $I$ be an instance of $P$
that has an input $X=\{x_1,\cdots,x_i,\cdots, x_m\}$ with input size $|I|=m$.
We say stream $S$ is a dynamic parameterized stream if $S$
is  a stream of $\textsc{Insert}(x_i)$ (i.e., the insertion
of an element $x_i$) and $\textsc{Delete}(x_i)$ (i.e., the deletion of
an element $x_i$) operations applying to the underlying instance
$(I,k)$ of $P$.
\end{definition}

Now stream $S$ is not simply a permutation
$X'=\{x'_1,\cdots,x'_i,\cdots, x'_m\}$ for $x'_i\in X$ of an input
$X$, but rather a sequence of transactions that collectively define a
graph. We assume the size of stream $S$
is $|S|\le m^c$ for a constant $c$ which means  $\log|S|\le c\log m$
or asymptotically, $O(\log|S|)= O(\log m)$. We denote the time which corresponds to
the $i$-th update operation of $S$ by \textit{time} $i$. The $i$-th update operation
can be $\textsc{Insert}(x'_i)$ or $\textsc{Delete}(x'_i)$ for $x'_i\in X$
(note that we can perform $\textsc{Delete}(x'_i)$ only if $x'_i$ is
present at time $i-1$). At time $i$,
the input of instance $I$ is a subset $X'_i\subseteq X$ of inputs
which are, up to time $i$, inserted but not deleted.

We next define a \textit{promised} streaming model as follows.
Suppose we know for sure that at every time $i$ of a dynamic parameterized
stream $S$, the size of the vertex cover of underlying graph $G(V,E)$ (where $E$
is the set of edges that are inserted up to time $i$ but not deleted)
is at most $k$. We show that within the framework of the
promised streaming model we are able to develop a dynamic parameterized
streaming algorithm whose space usage matches the lower bound of
Theorem \ref{thm:vc:lower:bound} up to $\tilde{O}(1)$ factor.

We formulate a dynamic parameterized streaming algorithm within the framework of
the promised streaming model as follows.

\begin{definition}
({\sc Promised dynamic parameterized streaming algorithm (PDPSA)})
Let  $S$ be a promised dynamic parameterized stream, i.e., we are promised
that at every time $i$, there is  a solution for input $X'_i$
whose  size fulfills the condition corresponding to $k$.
Let  $\mathcal{A}$ be an algorithm  that computes a sketch for problem $P$
using $\tilde{O}(f(k))$-space in one pass over stream $S$.
Suppose at the end of stream $S$, i.e., time $|S|$, algorithm $\mathcal{A}$
in time $\tilde{O}(g(k))$ extracts,  from the sketch,
a solution for input $X'_{|S|}$ (of instance $I$) whose  size fulfills
the condition corresponding to $k$.
We say $\mathcal{A}$ is an $(f(k), g(k))$-PDPSA.
\end{definition}


Next, using the well-known connection (see, for example, Chapter 9 of~\cite{FG06})
between maximal matching and kernelization algorithms for parameterized Vertex Cover,
we show that kernels for matching can be implemented in data streams in small space, and
this in turn gives  a PDPSA for $VC(k)$ problem.
We summarize this main result in the following theorem and
we develop it in Section \ref{sec:promised:dynamic}.

\begin{theorem}
\label{thm:vc:max:k}
Suppose at every timestep
the size of the vertex cover of underlying graph $G(V,E)$ is at most $k$.
There exists a $(k^2,2^{2k^2})$-PDPSA for $VC(k)$ with probability $\ge 1-\delta/n^c$,
where $\delta<1$ and $c$ is a constant.
\end{theorem}

Our algorithm takes the novel approach of combining linear sketching
with sequential operations that depend on the current state of the
graph.
Prior work in sketching has instead only performed updates of sketches
for each stream update, and postponed insepecting them until the end
of the stream.

As a byproduct of this main theorem we have the following corollary.

\begin{corollary}
\label{cor:general:MM:dyn:promised}
Assume we are promised that a maximal matching of underlying graph $G(V,E)$
at every time $i$ of dynamic parameterized stream $S$ is of size
at most $k$ for $k\in \NATURAL$.
Then, there exists a dynamic algorithm that maintains
a maximal matching of graph $G(V,E)$ using $\tilde{O}(k^2)$ space.
The update time (i.e., the time to maintain the sketch)
and query time (i.e., the time to maintain a maximal matching using the sketch)
of this algorithm are worst-case $\tilde{O}(k)$. For $k=\tilde{O}(\sqrt{n})$,
this gives a dynamic  algorithm for maximal matching whose space,
worst-case update and query times
are $\tilde{O}(n)$, $\tilde{O}(\sqrt{n})$ and $\tilde{O}(\sqrt{n})$, respectively.
\end{corollary}


%
Finally, we formulate a dynamic parameterized streaming algorithm without any promise as follows.

\begin{definition}
({\sc Dynamic parameterized streaming algorithm (DPSA)})
\label{defn:dpsa}
Let $S$ be a dynamic parameterized stream $S$.
Let  $\mathcal{A}$ be an algorithm  that computes a sketch for problem $P$
using $\tilde{o}(m)\cdot f(k)$-space and with one pass over stream $S$.
Suppose at the end of stream $S$, i.e., time $|S|$, algorithm $\mathcal{A}$
in time $\tilde{o}(m)\cdot g(k)$ extracts,  from the sketch,
a solution for input $X'_{|S|}$ whose  size fulfills
the condition corresponding to $k$  or reports that such a solution does not exist.
We say $\mathcal{A}$ is an $(\tilde{o}(m)\cdot f(k),\tilde{o}(m)\cdot g(k))$-DPSA.
\end{definition}

We state our result on the DPSA (without any promise) for Vertex Cover
and prove it in Section~\ref{app:vc:dpsa}.

\begin{theorem}
\label{THM:VC:DPSA}
Let $S$ be a dynamic parameterized stream of  insertions and deletions of edges of an underlying graph $G$.
There exists a randomized $(\min(m,nk), \min(m,nk)+2^{2k^2})$-DPSA for $VC(k)$ problem.
\end{theorem}

For graphs which are not  sparse (i.e., $m>O(nk)$) the algorithm of Theorem \ref{THM:VC:DPSA}
gives $(\tilde{o}(m)\cdot f(k),\tilde{o}(m)\cdot g(k))$-DPSA for $VC(k)$.
The space usage of PDPSA of Theorem \ref{thm:vc:max:k} matches
the lower bound of Theorem \ref{thm:vc:lower:bound}. On the other hand,
there is a gap between space bound $\tilde{O}(nk)$ of DPSA of Theorem \ref{THM:VC:DPSA}
and lower bound $\Omega(k^2)$ of Theorem \ref{thm:vc:lower:bound}.
We conjecture that the lower bound for the space usage of any
(randomized) DPSA for $VC(k)$ problem is indeed $\Omega(nk)$.


\subsection{Related Work}

The question of finding maximal and maximum cardinality matchings has
been heavily studied in the model of (insert-only) graph streams.
A greedy algorithm trivially obtains a maximal matching (simply store
every edge that links two currently unmatched nodes); this can also be
shown to be a 0.5-approximation to the maximum cardinality matching~\cite{Feigenbaum:Kannan:McGregor:Suri:Zhang:05}.
By taking multiple passes over the input streams, this can be improved
to a $1-\epsilon$ approximation, by finding augmenting paths with
successive passes~\cite{McGregor:05,McGregor:09}.

Subsequent work has extended to the case of weighted edges (when a
maximum weight matching is sought), and reducing the number of passes
to provide a guaranteed approximation~\cite{Eggert:Kliemann:Srivastav:09,Eggert:Kliemann:Munstermann:Srivastav:12}.
While approximating the size of the vertex cover has been studied in other sublinear models, such as
sampling~\cite{Parnas:Ron:07,Onak:Ron:Rosen:Rubinfeld:12}, we are not
aware of prior work that has addressed the question of finding a
vertex cover over a graph stream.
Likewise, parameterized complexity has not been explicitly studied in
the streaming model, so we initiate it here.

The model of dynamic graph streams has recently received much attention,
due to breakthroughs by Ahn, Guha and
McGregor~\cite{AGM12a,Ahn:Guha:McGregor:12PODS}.
Over two papers, they showed the first results for a number of graph
problems over dynamic streams, including
determining connected components,
testing bipartiteness,
minimum spanning tree weight and building a sparsifier.
They also gave multipass algorithms for maximum weight matchings
and spanner constructions.
This has provoked much interest into what can be computed over dynamic
graph streams.


\paragraph{Outline.}
Section~\ref{sec:prelims} provides background on techniques for
kernelization of graph problems, and on streaming algorithms for
building a sketch to recover a compact set.
Our results on PSA and DPSA are stated in Section~\ref{sec:psa:vc}
and Section~\ref{app:vc:dpsa}, respectively.
Section~\ref{sec:promised:dynamic} is the most involved, as it addresses the most
difficult dynamic case in the promised model.


\section{Preliminaries}
\label{sec:prelims}

In this section, we present the definitions of streaming model and
the graph sketching that we use.


\paragraph{Streaming Model.}
%
Let $S$ be a stream of insertions (or similarly, insertions and deletions) of edges of an underlying graph
$G(V,E)$. We assume that vertex set $V$ is fixed and given, and the size of $V$
is $|V|=n$. We assume that the size of stream $S$ is $|S|\le n^c$ for some large
enough constant $c$ so that we may assume that
$O(\log|S|)=O(\log n)$.
Here $[x]=\{1,2,3,\cdots,x\}$ when $x\in \NATURAL$.
Throughout the paper we denote a failure probabilities by $\delta$, and
approximation parameters by $\epsilon$.

We assume that there is a unique numbering for the vertices in
$V$ so that we can treat $v \in V$ as a unique number $v$ for $1 \le v \le n=|V|$.
We denote an undirected edge in $E$ with two endpoints $u,v\in V$ by
$(u,v)$. The graph $G$ can
have at most ${n \choose 2} = n(n-1)/2$ edges.
Thus, each edge can also be thought of as referring to a unique number
between 1 and ${n \choose 2}$.

%
At the start of stream $S$, edge set $E$ is an empty set.
We assume in the course of stream $S$, the maximum size of $E$ is a
number $m$, i.e., $m'=|E|\le m$. Counter $m'$ stores the current number of
edges of stream $S$, i.e., after every insertion we increment $m'$ by one and
after every deletion we decrement $m'$ by one.

Let $M$ be a maximal matching that we maintain for stream $S$.
Edges in $M$ are called \textit{matched} edges; the other edges are \textit{free}.
If $uv$ is a matched edge, then $u$ is the \textit{mate} of $v$ and $v$ is the
\textit{mate} of $u$. Let $V_M$ be the vertices of $M$ and
$\overline{V}_M=V\backslash V_M$. A vertex $v$ which is in $V_M$ is called
a \textit{matched} vertex, otherwise, i.e., if $v\in\overline{V}_M$, $v$ is called
an \textit{exposed} vertex.

The \textit{neighborhood} of a vertex $u\in V$ is defined as
$\mathcal{N}_u=\{v\in V: uv\in E\}$. Hence the degree of a vertex
$u\in V$ is $d_u=|\{uv\in E\}|=|\mathcal{N}_u|$.
We split the neighborhood of $u$ into the set
of matched neighbors of $u$, $\mathcal{N}_u \cap V_M$, and the set
of exposed neighbors of $u$, i.e.,
$\mathcal{N}_u \setminus V_M$.


%
\paragraph{Oblivious Adversarial Model.}
We work in the \textit{oblivious adversarial model} as is common for analysis of randomized data structures
such as universal hashing \cite{CW77}.
This model has been used in a series of papers on
dynamic maximal matching and dynamic connectivity problems: see for example
\cite{OR10, BGS11, KKM13, NS13}.
The model allows the adversary to know all the edges in the graph
$G(V,E)$ and their arrival order, as well as the algorithm to be
used.
However, the adversary is not aware of the random bits used by the
algorithm, and so cannot choose updates adaptively in response to the
randomly guided choices of the algorithm.
This effectively means that we can assume that the adversary prepares
the full input (inserts and deletes) before the algorithm runs.


\paragraph{$k$-Sparse Recovery Sketch and Graph Sketching.}
We first define an $\ell_0$-Sampler as follows.
\begin{definition}[$\ell_0$-Sampler \cite{FIS05, mw10}]
\label{def:l0}
Let $0< \delta <1$ be a parameter.
Let $S= (a_1,t_1), \cdots, (a_i,t_i), \cdots$ be a stream of updates of an underlying vector $x\in \REAL^n$
where $a_i\in[n]$ and $t_i\in \REAL$. The $i$-th update $(a_i,t_i)$
updates  the $a_i$-th element of $x$ using
$x[a_i]=x[a_i]+t_i$. A $\ell_0$-sampler algorithm for $x\neq 0$ returns FAIL with probability at most $\delta$.
Else, with probability $1-\delta$, it returns an element $j\in [n]$ such that the probability that $j$-th element
is returned is $\Pr{j}= \frac{|x_j|^0}{\ell_0(x)}$.
\end{definition}

Here, $\ell_0(x)=(\sum_{i\in[n]} |x_i|^0)$ is the
(so-called) ``$0$-norm'' of $x$ that counts the number of non-zero entries.

\begin{lemma}[\cite{JST11}]
\label{lem:l0:sampling}
Let $0< \delta <1$ be a parameter.
There exists a linear sketch-based algorithm for $\ell_0$-sampling
using $O(\log^2 n\log\delta^{-1})$ bits of space.
\end{lemma}

The concepts behind sketches for $\ell_0$-sampling can be generalized
to draw $k$ distinct elements from the support set of $x$:

\begin{definition}[$k$-sample recovery]
\label{def:ksparse}
A $k$-sample recovery algorithm recovers
$\min(k, \|x\|_0)$ elements from $x$ such that sampled index $i$ has
$x_i \neq 0$ and is sampled uniformly.
\end{definition}

Constructions of $k$-sample recovery mechanisms are known which
require space $\tilde{O}(k)$ and fail only with probability
polynomially small in $n$~\cite{Barkay:Porat:Shalem:12}.
We apply this algorithm to the neighborhood of vertices: for each
node $v$, we can maintain an instance of the $k$-sample recovery sketch (or algorithm) to the
vector corresponding to the row of the adjacency matrix for $v$.
Note that as edges are inserted or deleted, we can propagate these to
the appropriate $k$-sample recovery algorithms, without needing
knowledge of the full neighborhood of nodes.


Specifically,
let $a_1,\cdots,a_v,\cdots,a_n$ be the rows of the adjacency matrix of
$G$, $\mathcal{A}_G$, where
$a_v$ encodes the neighborhood of a vertex $v\in V$. We define the sketch of $\mathcal{A}_G$ as follows.
Let $S$ be a stream of insertions and deletions of edges to an underlying graph $G=(V,E)$.
We sketch each row $a_u$ of $\mathcal{A}_G$ using the sketching matrix of
Lemma \ref{lem:l0:sampling}. Let us denote this sketch by $S_{u}$.
Since sketch $S$ is linear, the following operations can be done in the sketch space.

\begin{itemize}
\item \textsc{Query}$(S_u)$:  This operation queries sketch $S_u$ to
find a uniformly random neighbor of vertex $u$. Since $S_u$ is a $k$-sample
recovery sketch, we can query up to $k$ uniformly random neighbors of vertex $u$.
\item \textsc{Update}$(S_u, \pm(u,v))$:
This operation updates the sketch of a vertex $u$.
In particular, operation \textsc{Update}$(S_u, (u,v))$
means that edge $(u,v)$ is added to  sketch $S_u$.
And, operation \textsc{Update}$(S_u, -(u,v))$
means that edge $(u,v)$ is deleted from  sketch $S_u$.
\end{itemize}


\section{Parameterized Streaming Algorithm (PSA) for $VC(k)$}
\label{sec:psa:vc}

In this section, we give a $(k^2, 2^{2k^2})$-PSA for $VC(k)$ along with a matching $\Omega(k^2)$ lower bound for the space complexity of $VC(k)$. First, we review the kernelization algorithm of Buss and Goldsmith~\cite{BG93}
since we use it in our PSA for $VC(k)$.

\subsection{Kernel for $VC(k)$}
\label{sec:kernel:VC}
%
Let $(G,k)$ be the original instance of the problem which is initialized by graph
$G=(V, E)$ and parameter $k$. Let $d_v$ denote the degree of $v$ in $G$.
While one of the following rules can be applied, we follow it.

\begin{enumerate}
 \item \underline{There exists a vertex $v\in G$ with $d_v > k$}: Observe that if we do not include $v$ in the vertex cover, then we must include all of $\mathcal{N}_{v}$. Since $|\mathcal{N}_v|=d_v >k$, we must include $v$ in our vertex cover for now. Update $G\leftarrow G\setminus \{v\}$ and $k\leftarrow k-1$.

 %
 %
 \item \underline{There is an isolated vertex $v\in G$}: Remove $v$ from $G$, since
 $v$ cannot cover any edge.

\end{enumerate}

If neither of above rules can be applied, then we look at the number of edges of $G$. Note that the maximum degree of $G$ is now $\leq k$. Hence, if $G$ has a vertex cover of size $\leq k$, then the maximum number of edges in $G$ is $k^2$. If $|E|> k^2$, then we can safely answer NO. Otherwise we now have a kernel graph $G=(V, E)$ such that $|E|\leq k^2$. Since $G$ does not have any isolated vertex, we have $|V|\leq 2|E|\leq 2k^2$. Observe that we obtain the kernel graph $G$ in polynomial time.
%
%
%
%
%
%
%

Now we show how to obtain an FPT algorithm for Parameterized Vertex Cover using the above kernelization: Enumerate all vertex subsets of
$G$ of size $k$, and checks whether any of them forms a vertex cover. The number of such subsets is $\binom{2k^2}{k}=2^{O(k^3)}$,
and checking whether a given subset is a vertex cover can be done in polynomial time.
We answer YES if any of the subsets of size $k$ forms a vertex cover, and NO otherwise. After obtaining the kernel graph in polynomial time, the running time of this algorithm is $2^{O(k^3)}\cdot n^{O(1)}$.

\subsection{$(k^2, 2^{2k^2})$-PSA for $VC(k)$}
\label{subsec:psa-vc}

We now prove Theorem~\ref{THM:VC:INSERTION}, which is restated below:


\begin{reptheorem}{THM:VC:INSERTION}
Let $S$ be a stream of  insertions of edges of an underlying graph $G$. Then
there exists a \emph{deterministic} $(k^2,2^{2k^2})$-PSA for $VC(k)$ problem.
\end{reptheorem}
\begin{proof}
The proof is divided into three parts: first we describe the algorithm, analyze its complexity and then show its correctness.

\paragraph{Algorithm.} Let $S$ be a stream of insertions of edges of an underlying graph $G(V,E)$.
We maintain a maximal matching $M$ of stream $S$ in a greedy fashion.
Let $V_M$ be the vertices of matching $M$. For every matched vertex $v$,
we also store up to $k$ edges incident on $v$ in a set $E_M$. If at a timestamp $i$ of stream $S$
we observe that $|M|>k$, we report that the size of any vertex cover of
$G=(V,E)$ is more than $k$ and quit.
At the end of stream $S$, we run the kernelization algorithm of Section~\ref{sec:kernel:VC}
on instance $(G_M=(V_M,E_M),k)$.


\paragraph{Complexity of the Algorithm.}
We observe that the space complexity of the algorithm is $O(k^2)$. In fact,
for each vertex $v\in V_M$ assuming $|M|\le k$ we keep at most $k$ edges, thus
we need space of at most $2k\cdot k=2k^2$. If $|M|>k$, as soon as the size
of the matching $M$ goes beyond $k$ we quit the algorithm and so in this case
we also use space of at most $2k\cdot k=2k^2$.
The query time of this algorithm
is dominated by the time to extract the vertex cover of $G_M$ (and
also of $G$) using the brute-force search algorithm (if one exists), which is $2^{O(2k^2)}$.


\paragraph{Correctness.}
We argue that
\begin{enumerate}
\item if the kernelization algorithm succeeds on
instance $(G_M=(V_M,E_M),k)$ and finds a vertex cover of size at most $k$ for $G_M$,
then that vertex cover is also a vertex cover of size at most $k$ for $G$.
\item On the other hand, if the kernelization algorithm reports that
instance $(G_M=(V_M,E_M),k)$ does not have a vertex cover of size at most $k$,
then instance $(G=(V,E),k)$ does not have a vertex cover of size at most $k$.
\end{enumerate}

First, note that trivially, any matching provides a lower bound on the
size of the vertex cover, and hence we are correct to reject if
$|M|>k$.

Otherwise, i.e., if $|M|\le k$, we write $d_v$ and $d'_v$ for the degree of $v$
in $G$ and $G_M$, respectively. We follow rules of the kernelization algorithm
on $G$ and $G_M$ in lockstep. Observe that since every edge $e\in E$ is incident
on at least one matched vertex $v\in V_M$, when an edge
$(u,v)\in E$ is not stored in $E_M$ it is in one of the following
cases.
\begin{enumerate}
\item $u\in V_M$ and $v\in V_M$:
Then, we must have $d'_u>k$ and $d'_v>k$ which means that $d_u>k$ and $d_v>k$.
\item Only $u\in V_M$:
Then, we must have $d'_u>k$ which means that $d_u>k$.
\item Only $v\in V_M$:
Then, we must have $d'_v>k$ which means that $d_v>k$.
\end{enumerate}


Now, let us consider a set $X=\{v_{k},v_{k-1},\cdots,v_r\}$ (for $r\ge 0$) of vertices
that Rule $(1)$ of the kernelization algorithm for $G_M$ removes. According to Rule
$(1)$, for a vertex $v_{k'}\in X$ (for $k\ge k'\ge r$) with $d'_{v_{k'}} > k'$, we
remove $v_{k'}$ and all edges incident on $v_{k'}$ from $G_M$ and decrease $k'$ by one.
Note that $d'_{v_{k'}} > k'$ iff $d_{v_{k'}} > k'$. This is due to the fact that,
before we remove vertex $v_{k'}$ from $G_M$,
we have removed only  those neighbors of $v_{k'}$ that are matched and the number of
such vertices is less than $k-k'$. Thus, Rule $(1)$ of the kernelization algorithm can be applied
on $G$ and we remove $v_{k'}$ and all edges
incident on $v_{k'}$ from $G$ and decrease $k'$ by one.


Next we consider Rule $(2)$. Assume in one step of the kernelization algorithm for $G_M$,
we have an isolated vertex $v\in G_M$. Observe that those neighbors of $v$ that we
have removed using Rule $(1)$ (before vertex $v$ becomes isolated) are all matched vertices
and the number of such vertices is less than $k$.
Moreover, $v$ never had any neighbor in $V\backslash V_M$ otherwise, $v$ is not isolated.
Thus, if $v$ has a neighbor $u$ in the remaining vertices of $V_M$,
edge $(u,v)$ must be in $E_M$ as we store up to $k$ edges incident on $v$ in
set $E_M$ which means $v$ is not isolated in $G_M$ and that is in contradiction
to our assumption that $v$ is isolated in $G_M$.
Since we run the kernelization algorithm on $G_M$ and
on $G$ for the vertices in set $X$, the same thing happens for $G$, i.e.,
$v$ in $G$ is also isolated. So, using Rule $(2)$, $v$ is removed from
$G_M$ iff $v$ is removed from $G$.


 Now assume neither Rule $(1)$ nor Rule $(2)$ can be applied for $G_M$,
 but the number of edges in $E_M$ is more than $k'^2$. The same thing must
 happen for $E$. Therefore, $G_M$ and $G$ do not have a vertex cover of size
 at most $k$.

If none of the above rules can be applied for $G_M$, we have a kernel $(G_M,k')$
such that $|V_M|\le 2k$ and $|E_M|\le k'^2\le k^2$. Now observe that after removal
of all vertices of $X$ and their incident edges from $G$, for every remaining vertex $v$
in $G_M$, $d_v\le k'$; otherwise $d_v> k'$ and $d'_v> k'$; so we can apply Rule $(1)$ which is
in contradiction to our assumption that none of the above rules can be applied for $G_M$.
Therefore,  kernel $(G_M,k')$ is also a kernel for $(G,k')$ and this proves the correctness
of our algorithm.
\end{proof}

\subsection{$\Omega(k^2)$ Lower Bound for $VC(k)$}
\label{sec:vc:bound}

Next, we prove Theorem~\ref{thm:vc:lower:bound} which is restated below:

\begin{reptheorem}{thm:vc:lower:bound}
Any (randomized) PSA for the $VC(k)$ problem requires $\Omega(k^2)$ space.
\label{thm:lower-bound-vc}
\end{reptheorem}
\begin{proof}

We will reduce from the \textsc{Index} problem in communication complexity:

\begin{center}\framebox{\begin{minipage}{0.8\columnwidth}
\underline{\textsc{Index}}\\
\emph{Input}: Alice has a string $X\in \{0,1\}^n$ given by
$x_{1}x_{2}\ldots x_n$. Bob has an index $\iota\in [n]$\\
\emph{Question}: Bob wants to find $x_\iota$, i.e., the $\iota^{th}$ bit of $X$.
\end{minipage}}\end{center}

It is well-known that there is a lower bound of $\Omega(n)$ bits
in the one-way randomized communication model for Bob to compute $x_i$~\cite{nisan}.
We assume an instance of the index problem where $n$ is a perfect
square, and let $k=\sqrt{n}$. Fix a canonical mapping from $[n]\rightarrow [k] \times
[k]$.
Consequently we can interpret the bit string as an adjacency matrix for a
bipartite graph with $k$ vertices on each side.

From the instance of \textsc{Index}, we construct an instance $G_X$
of Vertex Cover.
Assume that Alice has an algorithm which solves the $VC(k)$ problem
using $f(k)$ bits.
For each $i \in [k]$, we have vertices,
$v_i,v'_i, v''_i$, and $w_i, w'_i, w''_i$.
First, we insert the edges corresponding to the edge interpretation of
$X$ between nodes $v_i$ and $w_j$:
  for each $i, j \in [k]$, Alice adds the edge $(v_i, w_j)$ if
  the corresponding entry in $X$ is 1.
Alice then sends the memory contents of her algorithm to Bob, using
$f(k)$ bits.

Bob has the index $\iota \in [n]$, which he interprets as $(I,J)$ under
the same canonical remapping to $[k] \times [k]$.
He receives the memory contents of the algorithm, and proceeds to add
edges to the instance of vertex cover.
For each $i \in [k], i \neq I$, Bob adds the edges $(v_i,
v'_i)$ and $(v_i, v''_i)$.
Similarly, for each $j \in [k], j \neq J$, Bob adds the edges
$(w_j, w'_j)$ and $(w_j, w''_j)$.


\junk{
It follows that there is a lower bound of $\Omega(k^2)$ bits for Bob
to compute the matrix entry $M_{I,J}$. The next theorem shows that
if there is a single pass streaming algorithm, say $\mathcal{A}$,
which solves the $VC(k)$ problem in $f(k)$ space,
then there is a protocol for the \textsc{Matrix Index} instance with $f(k)$ bits.
Hence the lower bound of $\Omega(k^2)$ transfers to the $VC(k)$ problem as well.

Given an instance $M_{k\times k}$ of \textsc{Matrix Index},
we construct an instance $G_M$ of Vertex Cover as follows:
\begin{itemize}
\item For each $i\in [k]$ introduce three vertices $v_i, v'_i, v''_i$
\item For each $j\in [k]$ introduce three vertices $w_i, w'_i, w''_i$
\item For each $i,j\in [k]$ Alice adds the edge $v_i - w_j$ if and only if $M(i,j)=1$
\item Alice then sends this edge set to Bob using $f(k)$ bits
\item Bob knows the index $(I,J)$
\item For each $i\in [k], i\neq I$ Bob adds the edges $v_i - v'_i$ and $v_i - v''_i$
\item For each $j\in [k], j\neq J$ Bob adds the edges $w_j - w'_j$ and $w_j - w''_j$
\end{itemize}
}

The next lemma shows that finding the minimum vertex cover of $G_X$
allows us to solve the corresponding instance of \textsc{Index}.

\begin{lemma}
The minimum size of a vertex cover of $G_X$ is $2k-1$ if and only if
$x_\iota = 1$.
\label{thm:lower-bound-vc-insertion}
\end{lemma}
\begin{proof}
Suppose $x_\iota = 0$.
Then it is easy to check that the set $\{v_i : i\in [k], i\neq I\}\cup \{w_j\ |\ j\in [k], j\neq J\}$
forms a vertex cover of size $2k-2$ for $G_X$.

Now suppose
$x_\iota = 1$,
and let $Y$ be a minimum vertex cover for $G_X$.
For any $i\in [k], i\neq I$ the vertices $v'_i$ and $v''_i$ have degree one in $G_X$.
Hence, without loss of generality, we can assume that $v_i\in Y$.
Similarly, $w_j\in Y$ for each $j\in [k], j\neq J$.
This covers all edges except $(v_{I},w_{J})$.
To cover this we need to pick one of $v_I$ or $w_J$,
which shows that $|Y|=2k-1$.
\end{proof}
Thus, by checking whether the output of $\mathcal{A}$ on the instance
$G_X$ of $VC(k)$ is $2k-1$ or $2k-2$, Bob can determine the index
$x_\iota$.
The total communication between Alice and Bob was $O(f(k))$ bits,
and hence we can solve the \textsc{Index} problem in $f(k)$ bits. Recall that the
lower bound for the \textsc{Index} problem is $\Omega(n)=\Omega(k^2)$, and hence we have $f(k)=\Omega(k^2)$.
\end{proof}

\begin{corollary}
Let $1>\epsilon>0$. Any (randomized) PSA that approximates $VC(k)$ within a relative error of
$\epsilon$ requires $\Omega(\frac{1}{\epsilon^2})$ space.
\end{corollary}

\begin{proof}
Choose $\epsilon=\frac{1}{2k}$. Theorem~\ref{thm:lower-bound-vc-insertion}
shows that the relative error is at most $\frac{1}{2k-1}$, which is less than $\epsilon$.
Hence finding an approximation within $\epsilon$ relative error amounts to
finding the exact value of the vertex cover. The lower bound of $\Omega(k^2)$
from Theorem~\ref{thm:lower-bound-vc-insertion} translates to $\Omega(\frac{1}{\epsilon^2})$ here.
\end{proof}


\section{Promised Dynamic Parameterized Streaming Algorithm (PDPSA) for $VC(k)$}
\label{sec:promised:dynamic}
In this section we prove Theorem~\ref{thm:vc:max:k} which is restated here.
We let $c$ be a constant so that for the length of stream $S$ we have $|S|\le n^c$.

\begin{reptheorem}{thm:vc:max:k}
Assume that at every timestep
the size of the vertex cover of underlying graph $G(V,E)$ is at most $k$.
There exists a $(k^2,2^{2k^2})$-PDPSA for $VC(k)$ problem with probability at least $1-\delta/n^c$,
where $\delta<1$ and $c$ is a constant.
\end{reptheorem}


\subsection{Outline}
We develop a streaming algorithm that maintains a maximal matching of underlying graph $G(V,E)$
in a streaming fashion.
At the end of stream $S$ we run the kernelization algorithm of Section~\ref{sec:kernel:VC}
on the maintained maximal matching.
Our data structure to maintain a maximal matching $M$ of stream $S$ consists of two parts.

First, for each matched vertex $u$, we maintain an $x$-sample recovery sketch
$S_u$ of its incident edges, where
$x$ is chosen to be $\tilde{O}(k)$.
Insertions of new edges are relatively easy to handle: we update the
matching with the edge if we can, and update the sketches if the new
edge is incident on matched nodes.
The difficulty arises with deletions of edges: we must try to ``patch
up'' the matching, so that it remains maximal, using only the stored
information, which is constrained to be $O(k^2)$.
The intuition behind our algorithm is that, given the promise, there
cannot be more than $k$ matched nodes at any time.
Therefore, keeping $\tilde{O}(k)$ information about the
neighborhood of each matched node can be sufficient to identify any
adjacent unmatched nodes with which it can be paired if it becomes
unmatched.
However, this intuition requires significant care and case-analysis to
put into practice.
The reason is that we need some extra book-keeping to record where
information is stored, since nodes are entering and leaving the
matching, and we do not necessarily have access to the full
neighborhood of a node when it is admitted to the matching.
Nevertheless, we show that additional book-keeping information of size
$O(k^2)$ is sufficient for our purposes, allowing us to meet the
$O(k^2)$ space bound.

This book-keeping comes in the form of another
data structure $\mathcal{T}$, that
stores a set of edges $(u,v)$ such that both endpoints are matched
(not necessarily to each other),
and $(u,v)$ has been inserted into sketches $S_u$ and $S_v$, but not deleted from them.
The size of $\mathcal{T}$ is clearly $O(k^2)$.
To implement $\mathcal{T}$, we can adopt any fast dictionary data
structure (AVL-tree, red-black tree, or hash-tables).

The update at a time $t$ is either the insertion or the deletion of
an edge $(u,v)$ for $1\le t\le |S|$ where $|S|\le n^c$ is
the length of stream $S$.
We continue our outline of the algorithm by describing the behavior
in each case informally, with the formal details spelled out in
subsequent sections.

\paragraph{Insertion of an Edge $(u,v)$ at Time $t$.}
When the update at time $t$ is insertion of an edge $(u,v)$
two cases can occur.
The first case is if at least one of $u$ and $v$ is matched,
we insert edge $(u,v)$ to the sketches of those vertices (from $u$ and $v$)
which are matched.
If both $u$ and $v$ are matched, we also insert edge $(u,v)$
to $\mathcal{T}$.

The second case occurs if both vertices $u$ and $v$ are exposed.
We add edge $(u,v)$ to the current matching and to $\mathcal{T}$, and initialize sketches $S_u$
and $S_v$ by insertion of edge $(u,v)$ to $S_u$ and $S_v$.
However, we also need to perform some additional book-keeping updates
to ensure that the information is up to date.
Fix one of the nodes $u$.
There can be matched vertices, say $w\in V_M$, which are neighbors of $u$.
If previously an edge $(w, u)$ arrived while $u$ was not in the
matching, then 
we inserted $(w,u)$ to
sketch $S_w$, but $(w,u)$ was not inserted to sketch $S_u$
as $u$ was an exposed vertex at that time.
If at some subsequent point $w$ becomes an exposed vertex
and the matching edge $(u,v)$ is deleted
then vertex $u$ must have the option of
choosing an unexposed vertex $w$ to be rematched.
For that, we need to ensure that some information about the edge
$(w,u)$ is accessible to the algorithm.


A first attempt to address this is to try interrogating each sketch
$S_w$ for all edges incident on $u$, say when $u$ is first added to
the matching.
However, this may not work while respecting the space bounds:
$w$ may have a large number of neighbors, much larger than the limit
$x$.
In this case, we can only use $S_w$ to recover a sample of the
neighbors of $w$, and $u$ may not be among them.

To solve this problem we must wait until $w$
has low enough degree that we can retrieve its complete neighborhood
from $S_w$.
At this point, we can use these recovered edges to update the sketches
of other matched nodes.
We use the structure $\mathcal{T}$ to track information about edges
on matched vertices that are already represented in sketches, to avoid
duplicate representations of an edge.
This is handled during the deletion of an edge, since this is the only
event that can cause the degree of a node $w$ to drop.


\paragraph{Deletion of an Edge $(u,v)$ at Time $t$.}
When
the update at time $t$ is deletion of an edge $(u,v)$,
we have three cases to consider.
The first case is if only one of vertices $u$ and $v$ is matched, we
delete edge $(u,v)$ from the sketch of that matched vertex.

The second case is if both $u$ and $v$ are matched vertices,
but $(u,v)\notin M$.  We want to delete edge $(u,v)$ from sketches
$S_u$ and $S_v$, but $(u,v)$ might not be represented in both these sketches.
We need to find out if $(u,v)$ has been inserted to $S_u$ and $S_v$,
or only to one of them.
This can be found from $\mathcal{T}$.
If $(u,v)\in \mathcal{T}$, edge $(u,v)$ has been inserted to both $S_u$ and $S_v$.
So, we delete $(u,v)$ from both sketches safely.
Otherwise, i.e., if $(u,v)\notin \mathcal{T}$, $(u,v)$ has been
inserted to the sketch of only one of $u$ and $v$.
Assume that this is $u$.
To discover this we define {\em timestamps} for matched vertices.
The timestamp $t_i$ of a matched vertex $u$ is the (most recent) time
when $u$ was matched.
We show that edge $(u,v)$ is only in sketch $S_u$ (not $S_v$)
if and only if $(u,v)\notin \mathcal{T}$ and $t_u<t_v$.
Therefore, if $t_u<t_v$, we delete $(u,v)$ from sketch $S_u$.
Otherwise, i.e., if $t_v<t_u$, we delete $(u,v)$ from sketch $S_v$.
Observe that if $t_u=t_v$, we have inserted $(u,v)$ to $S_u$ and
$S_v$ as well as $\mathcal{T}$.


The third case
is when $(u,v)\in M$.
We delete edge $(u,v)$ from sketches $S_u$ and $S_v$ as well as matching $M$ and $\mathcal{T}$.
To maintain the maximality of matching $M$ we need to see whether
we can rematch $u$ and $v$.
Let us consider $u$ (the case for $v$ is identical).
If $u$ has high degree,
we sample edges $(u,z)$ from sketch $S_u$.
Given the size of the sketch, we argue that there is high probability
of finding an edge to rematch $u$.
Meanwhile, if $u$ is has low degree, then we can recover its full
neighborhood, and test whether any of these can match $u$.
Otherwise,
 $u$ is an exposed vertex, and its sketch is deleted.
We also remove all edges incident on $u$ from $\mathcal{T}$.


\subsection{Notations, Data Structures and Invariants}

We now describe and prove the properties of this process in full.
We begin with notations, data structures and invariants.

\paragraph{Timestamp of a Vertex and an Edge.}
We define \textit{time $t$} corresponding to the $t$-th update
operation (insert or delete of an edge) in stream $S$.
We define the {\em timestamp of a matched vertex} as follows.
Let $u$ be a matched vertex at time $t$.
Let $t'\le t$ be the greatest time such that $u$ was unmatched before
time $t'$ and is matched in the interval $[t',t]$.
Then we say the timestamp $t_u$ of vertex $u$ is $t'$ and we set
$t_u=t'$.
If at time $t$, vertex $u$ is exposed we define $t_u=\infty$,
i.e. a value larger than any timestamp.

We define the {\em timestamp of an edge} as follows.
Let $E_t$ denote the set of edges present at time $t$, i.e. which have
been inserted without a corresponding deletion.
Let $t'\le t$ be the last time in which the edge $(u,v) \in E_t$ is inserted
but not deleted in the interval $[t',t]$.
Then we say the timestamp $t_{(u,v)}$ of edge $(u,v)$ is $t'$ and we
set $t_{(u,v)}=t'$.
If at timestamp $t$, edge $(u,v)$ is deleted we define $t_{(u,v)}=\infty$.


\paragraph{Low and High Degree Vertices.}
Let $x=8ck\log(n/\delta)$, for constant $c$ (where, we assume that
$|S| = O(n^c)$).
At time $t$ we say a vertex $u$ is a \textit{high-degree} vertex if $d_u>x$;
otherwise, if $d_u\le x$, we say $u$ is a \textit{low-degree} vertex.


\paragraph{Data Structures.}
For every matched vertex $u$, i.e., $u\in V_M$, we maintain an
$x$-sample recovery sketch $S_u$ of edges incident on $u$.
We also maintain a dictionary data structure $\mathcal{T}$
of size $O(k^2)$.
We assume the insertion, deletion and query times of
$\mathcal{T}$ are all worst-case $O(\log k)$.
At every time $t$, $\mathcal{T}$ stores edges $(u,v)$ for which
vertices $u$ and $v$ are matched at time $t$ (not necessarily
  to each other); and also
%
edge $(u,v)$ is represented in both sketches $S_u$ and $S_v$,
  i.e.
there is a time $t'\le t$ at which we  invoked
\textsc{Update}$(S_{u}, (u,v))$, but there is no time
in interval $[t',t]$ in which we have invoked
\textsc{Update}$(S_{u}, -(u,v))$, and symmetrically for $S_v$.


\paragraph{Sketched Neighbors of a Vertex.}
Let $u$ be a matched vertex at some time $t$, i.e., $u\in V_M$.
Recall that $\mathcal{N}_u=\{v\in V: (u,v)\in E_t\}$ is the full neighborhood of $u$ at time $t$.
Let $\mathcal{N}'_u \subseteq \mathcal{N}_u$
be the set of neighbors of $u$ that up to time $t$
are inserted to $S_u$ but not deleted from $S_u$,
that is for every vertex $v\in \mathcal{N}'_u$ we have invoked
\textsc{Update}$(S_{u}, (u,v))$ at a time $t'\le t$ but have not invoked
\textsc{Update}$(S_{u}, -(u,v))$ in time interval $[t',t]$.
We call the vertices in $\mathcal{N}'_u$ the sketched neighbors of vertex $u$.
Note that we can recover $\mathcal{N}'_u$ {\em exactly} when
$|\mathcal{N}'_u| < x$, per Definition~\ref{def:ksparse}.


\junk{
\begin{lemma}
\label{lem:low:degree:recover:set}
Let $t$ be a time of stream $S$.
Let $u$ be a low-degree matched vertex at time $t$ that is $u\in V_M$ and $d_u\le x$.
Let $\mathcal{N}'_u$ be the set of neighbors of $u$ that up to time $t$
are inserted to sketch $S_u$ but not deleted from $S_u$.
By querying $S_u$ and with probability at least $1-\frac{\delta}{2n^c}$,
we can recover $\mathcal{N}'_u$.
\end{lemma}

\begin{proof}
From Definition \ref{def:l0}, a $\ell_0$-sampler returns an element $i\in [n]$
with probability $\Pr{i}= \frac{|x_i|^0}{\ell_0(x)}$ and returns FAIL with
probability at most $\delta$. Using Lemma \ref{lem:l0:sampling}, there exists
a linear sketch-based algorithm for $\ell_0$-sampling using
$O(\log^2 n\log\delta^{-1})$ bits of space.

Sketch $S_u$ is a $x$-sample recovery sketch which means
we can recover $\min(x,d_u)$ of items (here edges) that are
inserted into sketch $S_u$.  We can think of $S_u$
as $x$ instances of a $\ell_0$-sampler. Note the in this way the space
to implement $S_u$ would be $x$ times the space to implement a $\ell_0$-sampler
which is $O(x\log^2 n\log\delta^{-1})$ bits of space.

Each one of these $x$ $\ell_0$-samplers
returns FAILS with probability at most $\delta$. Using a union bound the probability
that $S_u$ returns FAIL is $x\delta$. We replace $\delta$ by
$\frac{\delta}{2xn^c}$ for a constant $c$. Therefore, the probability
that sketch $S_u$ returns FAIL is $\frac{\delta}{2n^c}$. Therefore, the overall space
of $S_u$ would be
$O(x\log^2 n\log(xn^c/\delta))=O(cx\log^2 n(\log(n/\delta)+\log\log(n/\delta)))=
O(x\log^2 n\log(n/\delta))$
as $x=8ck\log(n/\delta)$ and $k\le n$.

Let us condition on the event that $S_u$ does not return FAIL
which happens with probability $1-\frac{\delta}{2n^c}$.
We query $S_u$ and it returns $\min(x,d_u)$ edges.
Since $u$ is a low degree vertex ($d_u=|\mathcal{N}_u|\le x$),
we have $|\mathcal{N}'_u|\le |\mathcal{N}_u|\le x$.
Therefore, we recover all edges incident on $u$
which are added to sketch $S_u$ up to time $t$
but are not deleted from $S_u$. Thus, we can recover $\mathcal{N}'_u$.
\end{proof}
}


\paragraph{Invariants.}
Recall that at every time $t$ of stream $S$, set $E_t$ is the set of edges which
are inserted but not deleted up to time $t$.
We develop a dynamic algorithm that at every time $t$ of stream $S$
maintains the following three invariants.

\begin{center}
\shadowbox{
\parbox{\columnwidth} {
\begin{itemize}
\item \textbf{Invariant 1:} For every edge $(u,v)\in E_t$ at time $t$ we have at least one of $v\in \mathcal{N}'_u$ or $u\in \mathcal{N}'_v$.
\item Let $(u,v)\in E_t$ be an edge at time $t$ such that $u,v\in V_M$.
At time $t$,
     \begin{itemize}
     \item \textbf{Invariant 2:}  $u\notin \mathcal{N}'_v$ {\bf iff } $t_u < t_v$ and $(u,v) \notin\mathcal{T}$.
     \item \textbf{Invariant 3:}   $v\in \mathcal{N}'_u$ and $u\in \mathcal{N}'_v$ {\bf iff } $(u,v) \in\mathcal{T}$.
     \end{itemize}
\end{itemize}
}}
\end{center}

Observe that these invariants imply that at any time
$|\mathcal{T}| < 2k^2$.
That is, since $\mathcal{T}$ only holds edges such that both ends are
matched, and we assume that the matching has at most $2k$ nodes, then
the number of edges can be at most ${ {2k} \choose 2} < 2k^2$.

\junk{
\begin{lemma}
\label{lem:size:tree}
Suppose Invariant $3$ holds at a time $t$.
Then, at time $t$, we have $|\mathcal{T}|\le (2k)^2$.
\end{lemma}

\begin{proof}
At every time of stream $S$, for the size of a maximal matching $M$ we have $|M|\le k$.
Thus $|V_M|\le 2k$.
Using Invariant $3$, every edge $(u,v)$ at time $t$ is in $\mathcal{T}$ if
$u$ and $v$ are both matched. Therefore, $|\mathcal{T}|\le \frac{2k(2k+1)}{2}\le (2k)^2$
for $k>0$.
\end{proof}
}



\subsection{Adding an Edge to Matching $M$}
The first primitive that we develop is Procedure
{\sf AddEdgeToMatching$((u,v),t)$}.
This procedure first adds edge $(u,v)$ to matching $M$ and data structure $\mathcal{T}$.
Then it inserts vertex $u$ to $V_M$, sets timestamp $t_u$ to the current time
$t$, and initializes sketch $S_u$ by inserting
edge $(u,v)$ to sketch $S_u$.
It also repeats these steps for $v$.
We invoke this procedure in Procedures {\sf Rematch$((u,v),t)$}
and {\sf Insertion$((u,v), t)$}.

\begin{center}
\fbox{
\parbox{0.95\columnwidth} {
\raggedright
\underline{\textbf{Insertion$((u,v), t)$}}
\vspace{-0.2cm}
\begin{enumerate}
\item  If $u \notin V_M$ and $v \notin V_M$, then {\sf AddEdgeToMatching$((u,v), t)$}.
\item  Else {\sf InsertToDS$((u,v)$)}.
\end{enumerate}
\vspace{-0.2cm}
}}
\end{center}


\begin{center}
\fbox{
\parbox{0.95\columnwidth} {
\underline{\textbf{AddEdgeToMatching$((u,v), t)$}}
\begin{enumerate}
\item Add edge $(u,v)$ to $M$ and $\mathcal{T}$.
\item For $z \in \{u,v\}$ 
    \begin{enumerate}
     \item $V_M \gets V_M \cup \{z\}$
     \item $t_z \gets t$
     \item Initialize sketch $S_z$ with \textsc{Update}$(S_{z}, (u,v))$.
    \end{enumerate}
\end{enumerate}
}}
\end{center}


\begin{lemma}
 \label{lem:timestamps:case:1}
Let $t$ be a time when we invoke Procedure {\sf AddEdgeToMatching$((u,v), t)$}.
Suppose before time $t$, Invariants $1$, $2$ and $3$ hold.
Then, Invariants $1$, $2$ and $3$ hold after time $t$.
\end{lemma}

\begin{proof}
Recall that $t_u$ is the last time $t'\le t$ such that $u$ before time $t'$ was
unmatched and is matched in the interval $[t',t]$.
Similarly, $t_v$ is the last time $t'\le t$ such that $v$ before time $t'$ was
unmatched and is matched in the interval $[t',t]$.

In Procedure {\sf AddEdgeToMatching$((u,v), t)$} we insert $(u,v)$ to sketches
$S_u$ and/or $S_v$ if the edge has not been inserted to these sketches.
So, at time $t$, Invariant $1$ for edge $(u,v)$ holds.
Since $(u,v)\in M$, nothing changes for Invariants $2$ and $3$.
Therefore, if Invariants $2$ and $3$ hold at time $t-1$,
they also hold at time $t$.
\end{proof}


\subsection{Maintenance of Data Structure $\mathcal{T}$}
To maintain data structure $\mathcal{T}$ at every time $t$ of stream $S$,
we develop two procedures to handle insertions and deletions to the
structure.
If $u$ and $v$ are  matched vertices,
Procedure {\sf InsertToDS($(u,v)$)} inserts edge $(u,v)$
to sketches $S_u$ and $S_v$ as well as to data structure $\mathcal{T}$.
If only one of $u$ and $v$ is matched, we insert $(u,v)$
to the sketch of the matched vertex.
We invoke this procedure upon insertion of an arbitrary edge $(u,v)$
inside Procedure {\sf Insertion$((u,v),t)$}.


\begin{center}
\fbox{
\parbox{0.95\columnwidth} {
\underline{\textbf{InsertToDS$((u,v))$}}
\begin{enumerate}
\item If $u\in V_M$ and $v\in V_M$ then
insert edge $(u,v)$ into $\mathcal{T}$.
\item If $u\in V_M$ then \textsc{Update}$(S_{u}, (u,v))$.
\item If $v\in V_M$ then \textsc{Update}$(S_{v}, (u,v))$.
\end{enumerate}
}}
\end{center}


\begin{lemma}
 \label{lem:timestamps:case:2}
Let $t$ be a time of stream $S$ when we invoke Procedure {\sf InsertToDS($(u,v)$)}.
Suppose before time $t$, Invariants $1$, $2$ and $3$ hold.
Then, Invariants $1$, $2$ and $3$ hold after time $t$.
\end{lemma}

\begin{proof}
First assume at time $t$ when we invoke Procedure {\sf InsertToDS($(u,v)$)},
vertices $u$ and $v$ are already matched.
In Procedure {\sf InsertToDS($(u,v)$)} we insert $(u,v)$ to sketches
$S_u$ and $S_v$ using \textsc{Update}$(S_{u}, (u,v))$)
and \textsc{Update}$(S_{v}, (u,v))$).
So, $u\in \mathcal{N}'_v$ and $v\in \mathcal{N}'_u$ and  Invariant $1$ holds.
Moreover, we insert $(u,v)$ to $\mathcal{T}$.
Therefore, Invariant $3$ holds.
Invariant $2$ also holds as
neither condition is true ($v\notin \mathcal{N}'_u$ and $(u,v)\notin\mathcal{T}$).

Next assume only vertex $u$ is matched.
We insert $(u,v)$ to sketch $S_u$, but not to $S_v$ and $\mathcal{T}$.
Since $v\in \mathcal{N}'_u$, Invariant $1$ is correct.
Invariant $2$ and $3$ are correct as $v$ is not matched at time $t$.
The case when only vertex $v$ is matched is symmetric.
\end{proof}


The second procedure is {\sf DeleteFromDS$((u,v))$}
which is invoked in Procedure {\sf Deletion$((u,v), t)$} when $(u,v)\notin M$.
There are three main cases to consider.
If $(u,v)\in\mathcal{T}$,
we delete $(u,v)$ from sketches $S_u$ and $S_v$
as well as data structure $\mathcal{T}$.
If not, we know that $(u,v)$  is only in one of $S_u$ and $S_v$.

If $t_u<t_v$ and both $u$  and $v$ are matched, we delete the edge from $S_u$,
otherwise, if $t_v<t_u$ and $u$ and $v$ are matched from $S_v$, we delete the edge from $S_v$.
If none of these cases occur, then only one of $u$ and $v$ is matched.
If the matched vertex is $u$, we delete $(u,v)$ from $S_u$.
Otherwise, we delete $(u,v)$ from $S_v$.


\begin{center}
\fbox{
\parbox{0.95\columnwidth} {
\underline{\textbf{Deletion$((u,v), t)$}}
\vspace{-0.2cm}
\raggedright
\begin{enumerate}
\item  If $(u,v)\in M$ then invoke {\sf Rematch$((u,v),t)$}
\item  Else invoke {\sf DeleteFromDS$((u,v))$}.
\item  Invoke {\sf AnnounceNeighborhood$(u)$} and {\sf AnnounceNeighborhood$(v)$}
\end{enumerate}
\vspace{-0.2cm}
}}
\end{center}


\begin{center}
\fbox{
\parbox{0.95\columnwidth} {
\underline{\textbf{DeleteFromDS$((u,v))$}}
\raggedright
\begin{enumerate}
\item If $(u,v)\in \mathcal{T}$ then
\begin{enumerate}
   \item \textsc{Update}$(S_{u}, -(u,v))$ and \textsc{Update}$(S_{v}, -(u,v))$.
   \item Remove $(u,v)$ from $\mathcal{T}$.
\end{enumerate}
\item Else if $t_u< t_v $ and $u,v\in V_M$ then \textsc{Update}$(S_{u}, -(u,v))$.
\item Else if $t_v< t_u$ and $u,v\in V_M$ then \textsc{Update}$(S_{v}, -(u,v))$.
\item Else if $u\in V_M$ and $v\notin V_M$ then \textsc{Update}$(S_{u}, -(u,v))$.
\item Else if $v\in V_M$ and $u\notin V_M$ then \textsc{Update}$(S_{v}, -(u,v))$.
\end{enumerate}
}}
\end{center}


\begin{lemma}
\label{lem:timestamps:case:3}
Assume Invariants $1$, $2$ and $3$ hold at time $t$ when Procedure
{\sf DeleteFromDS$((u,v))$} is invoked.
Then, Procedure {\sf DeleteFromDS($(u,v)$)} chooses the correct case.
\end{lemma}

\begin{proof}
First, we consider the case that  both $u$ and $v$ are matched vertices.
Since Invariant $3$ holds, we know that edge $(u,v)$
at time $t$ is in $\mathcal{T}$ if and only if
$v\in \mathcal{N}'_u$ and $u\in \mathcal{N}'_v$.
Procedure {\sf DeleteFromDS($(u,v)$)} searches for $(u,v)$ in $\mathcal{T}$.
If this finds $(u,v)$ in $\mathcal{T}$, we then know that $v\in \mathcal{N}'_u$
and $u\in \mathcal{N}'_v$. So, we can safely delete the edge from
sketches $S_u$ and $S_v$ and data structure $\mathcal{T}$.

On the other hand, if $(u,v)\notin \mathcal{T}$, we ensure that the edge
is in only one of $S_u$ and $S_v$.  Now, we can use the claim of Invariant
$2$ which says $u\notin \mathcal{N}'_v$ if and only if $t_u < t_v$ and $(u,v) \notin\mathcal{T}$.
We compare $t_u$ and $t_v$. If $t_u<t_v$, then $u\notin \mathcal{N}'_v$.
Recall that since Invariant $1$ holds, we know that at least one of $v\in \mathcal{N}'_u$
and $u\in \mathcal{N}'_v$ is correct.
Because $u\notin \mathcal{N}'_v$, we must have $v\in \mathcal{N}'_u$.
So deleting edge $(u,v)$ from sketch $S_u$ is the correct operation.
On the other hand, if $t_v<t_u$, then $v\notin \mathcal{N}'_u$ and so
edge $(u,v)$ is only in sketch $S_v$. Thus, deleting edge $(u,v)$ from
sketch $S_v$ is the correct operation.

Next we consider the case that only one of $u$ and $v$ is matched.
Let us assume $u$ is the matched vertex.
Since Invariant $1$ holds, we know that at least one of $v\in \mathcal{N}'_u$
and $u\in \mathcal{N}'_v$ is correct. Because $u$ is the matched vertex and
we maintain the sketch of matched vertices, $(u,v)$ has been inserted to sketch
$S_u$ that is $v\in \mathcal{N}'_u$. Therefore,
deleting edge $(u,v)$ from sketch $S_u$ is the correct operation.
The case when $v$ is the matched vertex is symmetric.
\end{proof}



\subsection{Announcement and Deletion of Neighborhood of a Vertex}
In this section we develop basic primitives for the announcement
and deletion of the neighborhood of a vertex.
Announcement is performed by
Procedure {\sf AnnounceNeighborhood$(u)$} which is
invoked in Procedure {\sf Deletion$((u,v), t)$}.
Suppose that node $u$ has low degree.
For every matched vertex $v\in \mathcal{N}'_u$,
we search for edge $(u,v)$ in $\mathcal{T}$. If $(u,v)\in\mathcal{T}$,
$(u,v)$ is in both $S_u$ and $S_v$ and no action is needed.
But if not, we insert edge $(u,v)$ into tree $\mathcal{T}$ as well as
sketch $S_v$.


\begin{center}
\fbox{
\parbox{0.95\columnwidth} {
\underline{\textbf{AnnounceNeighborhood$(u)$}}
\vspace{-0.1cm}
\begin{enumerate}
\item If $u\in V_M$ and $d_u\le x$, then
     \begin{enumerate}
     \item For every edge $(u,v)$ in sketch $S_u$
           \begin{enumerate}
            \item Add $v$ to set $\mathcal{N}'_u$.
           \end{enumerate}
     \item For every $v\in \mathcal{N}'_u \cap V_M$
               \begin{enumerate}
               \item 
If edge $(u,v)\notin \mathcal{T}$, then
insert $(u,v)$ to $\mathcal{T}$; \textsc{Update}$(S_{v}, (u,v))$.
               \end{enumerate}
     \end{enumerate}
\end{enumerate}
}}
\end{center}



We also introduce a deletion primitive in the form of Procedure {\sf
  DeleteNeighborhood$(u)$}.  This is
invoked in {\sf Rematch$((u,v), t)$} when
the matched edge $(u,v)$ is removed. 
The {\sf DeleteNeighborhood$(u)$} procedure is called on a node $u$ when 
all the following three conditions hold.
\begin{enumerate}
 \item The matched edge of matched vertex $u$ is deleted.
 \item Vertex $u$ is a low-degree vertex.
 \item Vertex  $u$ does not have any exposed neighbor.
\end{enumerate}

In this case, we need to delete $u$ from $V_M$ and delete incident edges on
$u$ from data structure $\mathcal{T}$ as Invariant $3$ for $u$ is not valid
anymore. More precisely, for a low-degree matched vertex whose neighborhood
are all matched we do as follows.

We recover all edges from the sketch $S_u$ (i.e. $\mathcal{N}'_u$).
For every edge $(u,v) \in \mathcal{N}'_u$,
we check to see if $(u,v)\in\mathcal{T}$.
If so, we know that $(u,v)$
is represented in both sketches $S_u$ and $S_v$.
We also delete $(u,v)$ from $\mathcal{T}$ as $u$ is not
matched and Invariant $3$ does not hold.
But if $(u,v)\notin\mathcal{T}$,
since Invariant $1$ holds we know that $(u,v)$ is inserted only in $S_u$ not in $S_v$.
Observe that since $u$ does not have
any exposed neighbor, vertex $v$ must be a matched vertex, and so
vertex $v$ has an associated sketch $S_v$.
Therefore, in order to fulfill Invariant $1$, we first insert $(u,v)$ to sketch $S_v$.
Finally, we delete the whole sketch $S_u$,
and remove $u$ from $V_M$.


\begin{center}
\fbox{
\parbox{0.95\columnwidth} {
\underline{\textbf{DeleteNeighborhood$(u)$}}
\begin{enumerate}
     \item For every edge $(u,v)$ in sketch $S_u$
           \begin{enumerate}
           \item If edge $(u,v)\in \mathcal{T}$, then
Remove $(u,v)$ from $\mathcal{T}$.
           \item Else \textsc{Update}$(S_{v}, (u,v))$.
           \end{enumerate}
\item Delete sketch $S_u$ and remove $u$ from $V_M$.
\end{enumerate}
}}
\end{center}


\COMMENTED{

\begin{figure*}[t]
\centering
\fbox{
\parbox{6.0in} {
\begin{tabular}{ l | r }
\parbox{2.9in} {
\underline{\textbf{AnnounceNeighborhood$(u)$}}
\begin{enumerate}
\item If $u\in V_M$ and $d_u\le x$, then
     \begin{enumerate}
     \item For every edge $(u,v)$ in sketch $S_u$
           \begin{enumerate}
            \item Add $v$ to set $\mathcal{N}'_u$.
           \end{enumerate}
     \item For every $v\in \mathcal{N}'_u \cap V_M$
           \begin{enumerate}
            \item If edge $(u,v)\notin \mathcal{T}$, then
               \begin{enumerate}
               \item Insert $(u,v)$ to $\mathcal{T}$.
               \item \textsc{Update}$(S_{v}, (u,v))$.
               \end{enumerate}
           \end{enumerate}
     \end{enumerate}
\end{enumerate}
}
\hspace{-0.5cm}
&
\parbox{3.1in} {
\vspace{-0.4cm}
\underline{\textbf{DeleteNeighborhood$(u)$}}
\begin{enumerate}
     \item For every edge $(u,v)$ in sketch $S_u$
           \begin{enumerate}
           \item If edge $(u,v)\in \mathcal{T}$, then
                    \begin{enumerate}
                    \item Remove $(u,v)$ from $\mathcal{T}$.
                    \end{enumerate}
           \item Else \textsc{Update}$(S_{v}, (u,v))$.
           \end{enumerate}
\item Delete sketch $S_u$ and remove $u$ from $V_M$.
\end{enumerate}
}
\end{tabular}
}}
\end{figure*}

}

\begin{lemma}
\label{lem:timestamps:case:4}
Let $t$ be a time when we invoke Procedure {\sf AnnounceNeighborhood$(u)$}.
Suppose $u$ is a low-degree matched vertex at time $t$.
Suppose before time $t$, Invariants $1$, $2$ and $3$ hold.
Then after time $t$, Invariants $1$, $2$ and $3$ hold.
\end{lemma}

\begin{proof}
Let $\mathcal{N}'_u$ be the set of neighbors of $u$ that up to time $t$
are inserted into sketch $S_u$ but not deleted from $S_u$.
Since $u$ at time $t$ is a low-degree vertex we can use
Definition~\ref{def:ksparse} to recover $\mathcal{N}'_u$ in its entirety.
We assume Invariants $1$, $2$ and $3$ hold before time $t$.
We prove that all three invariants continue to hold after invocation of
{\sf AnnounceNeighborhood$(u)$}.

Fix a matched neighbor $v$ of $u$ in $\mathcal{N}'_u$ that is
$v\in V_M\cap\mathcal{N}'_u$. In Procedure
{\sf AnnounceNeighborhood$(u)$} for $v$ we do the following.
If edge $(u,v)$ has not been already inserted in $\mathcal{T}$,
we insert edge $(u,v)$ to $\mathcal{T}$ and $S_v$.
So, now $v\in \mathcal{N}'_u$ and $u\in \mathcal{N}'_v$, and
$(u,v) \in \mathcal{T}$. Invariants $1$, $2$ and $3$ hold for $(u,v)$,
and continue to hold for all other edges.
\end{proof}

After processing this deletion, edge $(u,v)$ is no longer in $E_t$,
and so the invariants trivially hold in regard of this edge.
Meanwhile, for any other edge, if the invariants held before, then
they continue to hold, since the changes only affected edge $(u,v)$.


\begin{lemma}
\label{lem:timestamps:case:5}
Suppose before time $t$, Invariants $1$, $2$ and $3$ hold and we
invoke {\sf DeleteNeighborhood$(u)$} at time $t$.
Here we assume $u$ is a matched vertex whose neighbors are all matched, i.e.,
$\mathcal{N}_u\cap \overline{V}_M=\emptyset$.
Then after time $t$, Invariants $1$, $2$ and $3$ hold.
\end{lemma}

\begin{proof}
Let $(u,v)$ be an edge in sketch $S_u$.
Since we assume Invariant $1$ holds before time $t$, $(u,v)$ must be
inserted into at least one of $S_u$ and $S_v$. We know edge $(u,v)$ is in $S_u$.
Since Invariants $2$ and $3$ hold, we have one of the two following cases.

(i) If edge $(u,v)$ is also inserted to $S_v$, this means this edge must be in $\mathcal{T}$.
In {\sf DeleteNeighborhood$(u)$} if edge $(u,v)$ is in $\mathcal{T}$,
we delete the edge from $\mathcal{T}$ as well as sketch $S_u$.
As $(u,v)$ is still in $S_v$, Invariant $1$ after time $t$ holds.

(ii) Else, edge $(u,v)$ is not in $S_v$. Using Invariant $2$ this happens
if and only if $t_u < t_v$ and $(u,v)\notin \mathcal{T}$.
We want to delete all edges which are inserted to $S_u$ and delete sketch $S_u$.
Observe that since $u$ does not have
any exposed neighbor, vertex $v$ must be a matched vertex and so has
an associated sketch $S_v$.
We insert $(u,v)$ to sketch $S_v$, and
subsequently $S_u$ is deleted.
Therefore, Invariant $1$ still holds.

Finally, Invariants $2$ and $3$ hold after time $t$ since $u$
is not a matched vertex anymore.
\end{proof}


\subsection{Rematching Matched Vertices}
In this section we develop the last (and most involved) primitive, {\sf Rematch$((u,v),t)$}.
We invoke this procedure in Procedure {\sf Deletion$((u,v), t)$} when the matched edge
$(u,v)$ is deleted. We first delete edge $(u,v)$ from sketches $S_u$ and $S_v$ as well as data structure $\mathcal{T}$.
We also delete $(u,v)$ from current set $M$ of matched edges.
To see if we can rematch $u$ and $v$ to one of their exposed neighbors,
we perform the subsequent steps for $u$ (and then repeat for $v$).

If $u$ is a low degree vertex, by querying $S_u$ we recover $\mathcal{N}'_u$, i.e.,
the set of neighbors of $u$ that up to time $t$
are inserted into sketch $S_u$ but not deleted from $S_u$.
We then check whether there is an exposed vertex
$z\in \mathcal{N}'_u$. If so, we rematch $u$ to $z$.

But if there is no exposed vertex in $\mathcal{N}'_u$,
we announce $u$ as an exposed vertex.
We also remove sketch $S_u$ as $u$ is not a matched vertex anymore. Moreover,
we remove all incident edges on $u$ from $\mathcal{T}$ as our third invariant
does not hold anymore. Lemma \ref{lem:low:rematch} shows that in both cases,
the matching after invoking Procedure {\sf Rematch$((u,v),t)$} is maximal
if the matching before this invocation was maximal.

If $u$ is a high degree vertex, it samples an edge $(u,z)$ from sketch $S_u$.
In Lemma \ref{lem:high:rematch} we show that with high probability $z$
is an exposed vertex, so we rematch $u$ to $z$.
Therefore, if the matching before the invocation of
Procedure {\sf Rematch$((u,v),t)$} is maximal, the matching after this
invocation would be maximal as well.

\begin{center}
\fbox{
\parbox{0.95\columnwidth} {
\raggedright
\underline{\textbf{Rematch$((u,v),t)$}}
\begin{enumerate}
\item {\sf DeleteFromDS$((u,v))$}, remove $(u,v)$ from $M$, remove
  $u$, $v$ from $V_M$
\item For $w\in\{u,v\}$
\begin{enumerate}
\item If $d_w\le x$ then
      \begin{enumerate}
      \item For every edge $(w,z)$ in sketch $S_w$, add $z$ to set $\mathcal{N}'_w$.
      \item If there is an exposed  $z \in\mathcal{N}'_w$ then invoke {\sf AddEdgeToMatching$((w,z), t)$}.
      \item Else invoke {\sf DeleteNeighborhood$(\text{vertex } w)$}.
      \end{enumerate}
\item If $d_w> x$ then
      \begin{enumerate}
      \item Query edges $(w,z_1),\cdots, (w,z_{y})$ from sketch $S_w$ for $y=8c\log(n/\delta)$.
      \item If there is an exposed $z\in \{z_1,\cdots,z_{y}\}$ then invoke {\sf AddEdgeToMatching$((w,z), t)$}.
      \end{enumerate}
\end{enumerate}
\end{enumerate}
}}
\end{center}


\subsubsection{Analyzing Rematching of a Low-Degree Vertex.}

\begin{lemma}
\label{lem:low:rematch}
Let $u$ be a low-degree matched vertex at time $t$.
Assuming the matching $M$ before time $t$ is maximal,
then, after the invocation of Procedure {\sf Rematch$((u,v), t)$}, the matching $M$ is maximal.
The running time of Procedure {\sf Rematch$((u,v), t)$}
when $u$ is a low-degree vertex is $O(k\log^4(n/\delta))$.
\end{lemma}

\begin{proof}
Let $\mathcal{N}'_u$ be the set of neighbors of $u$ up to time $t$
that are inserted into sketch $S_u$ but not deleted from $S_u$.
From Definition~\ref{def:ksparse},
by querying $S_u$ and with probability at least
$1-\frac{\delta}{2n^c}$, we can recover $\mathcal{N}'_u$.
Observe that assuming Invariants $1$, $2$ and $3$ hold, we must have
$\mathcal{N}_u \backslash \mathcal{N}'_u \subseteq V_M$, that is, those neighbors of
$u$ that are not in $\mathcal{N}'_u$ at time $t$ must be matched.
Therefore, all exposed neighbors of $u$ must be in $\mathcal{N}'_u$.

Two cases can occur. The first is if there is an exposed vertex $z$ in $\mathcal{N}'_u$.
Then, Procedure {\sf Rematch$((u,v),t)$} will rematch $u$ using exposed vertex $z$.
The second is when all neighbors of $u$ are already matched. Since all neighbors of $u$
are matched, vertex $u$ cannot be matched to one of its neighbors and
so we announce $u$ as an exposed vertex and release its sketch $S_u$.
Therefore, assuming $M$ before time $t$ is maximal, $M$ after time $t$
would be maximal as well.

We next discuss the running time of Procedure {\sf Rematch$((u,v),t)$}
when $u$ is a low-degree vertex.
By properties of the sketch data structures,
the time to query $x$ sampled edges from sketch $S_u$ and construct set
$\mathcal{N}'_u$ is $O(x\log^2 n\log(n/\delta))$.
If the second case happens, since we assume at every time of stream $S$, $|M|\le k$,
we then have $d_u=|\mathcal{N}'_u|\le 2k$.

Recall that $\mathcal{T}$ is a data structure with at most $k^2$ edges
whose space is $O(k^2)$. The insertion, deletion and
search times of $\mathcal{T}$ are all worst-case $O(\log k)$.
In the second case, the main cost is to remove incident edges
on $u$ from $\mathcal{T}$.
For every neighbor $z\in \mathcal{N}'_u$ we search, in time $O(\log k)$,
if edge $(u,z)$ has been inserted into $\mathcal{T}$; so overall the deletion of incident
edges on $u$ from $\mathcal{T}$ is done in time $O(k\log k)=O(x\log k)$
as $|\mathcal{N}'_u|\le 2k$. Overall, the running time of Procedure {\sf Rematch$((u,v),t)$}
when $u$ is a low-degree vertex is $O(x\log^2
n\log(n/\delta))=O(k\log^4(n/\delta))$, as we set $x = O(k \log (n/\delta))$.
\end{proof}


\subsubsection{Analyzing Rematching of a High-Degree Vertex.}

\begin{lemma}
\label{lem:high:rematch}
Let $x=8ck\log(n/\delta)$ and $y=8c\log(n/\delta)$.
Let $u$ be a high degree vertex, i.e., $d_u> x$.
Suppose we query edges $(u,z_1),\cdots,(u,z_i),\cdots, (u,z_y)$ from sketch $S_u$.
The probability that there exists an exposed vertex $z\in \{z_1,\cdots,z_y\}$ is
at least $1-\delta/n^c$.
Further, the running time of Procedure {\sf Rematch$((u,v), t)$}
when $u$ is a high-degree vertex is $O(\log^4(n/\delta))$.
\end{lemma}

\begin{proof}
From Definition \ref{def:l0}, a $\ell_0$-sampler returns an element $i\in [n]$
with probability $\Pr{i}= \frac{|x_i|^0}{\ell_0(x)}$ and returns FAIL with
probability at most $\delta$. Using Lemma \ref{lem:l0:sampling}, there exists
a linear sketch-based algorithm for $\ell_0$-sampling using
$O(\log^2 n\log\delta^{-1})$ bits of space.

Sketch $S_u$ is a $x$-sample recovery sketch which means
we can recover $\min(x,d_u)$ items (here, edges) that are
inserted into sketch $S_u$.
We can think of $S_u$ as $x$ instances of a $\ell_0$-sampler.
Note that in this way the space
to implement $S_u$ would be $x$ times the space to implement a $\ell_0$-sampler
which is $O(x\log^2 n\log\delta^{-1})$ bits of space.
Each one of these $x$ $\ell_0$-samplers
returns FAIL with probability at most $\delta$.
Using a union bound the probability that $S_u$ returns FAIL is
$x\delta$.
We rescale the failure probability $\delta$ to $\frac{\delta}{2xn^c}$
for a constant $c$.
Therefore, the probability that sketch $S_u$ returns FAIL is
$\frac{\delta}{2n^c}$, and hence the overall space of $S_u$ is
$O(x\log^2 n\log(xn^c/\delta))=O(cx\log^2 n(\log(n/\delta)+\log\log(n/\delta)))=O(cx\log^2 n\log(n/\delta))$
as $x=8ck\log(n/\delta)$ and $k\le n$.

Let $(u,z_1),\cdots,(u,z_i),\cdots, (u,z_y)$ be the edges queried from sketch $S_u$ for $y=8c\log(n/\delta)$.
Note that the time to query $y$ edges from sketch $S_u$ is
$O(y\log^2 n\log(n/\delta))=O(\log^4(n/\delta))$.
Let us define event $NoFAIL$ if $S_u$ does not return FAIL.
Let us condition on event $NoFAIL$ which happens with probability
$\Pr{NoFAIL}\ge 1-\frac{\delta}{2n^c}$.

Fix a returned edge $(u,z_i)$.
Recall that $\mathcal{N}_u$ is the neighborhood
of $u$ that is, $\mathcal{N}_u=\{v\in V: (u,v)\in E_t\}$.
The number of matched vertices is at most $2k$,
i.e., $|V_M|\le 2k$. Thus, $|\mathcal{N}_u \cap V_M|\le 2k$ and
$|\mathcal{N}_u \backslash \mathcal{N}'_u|=|\mathcal{N}_u| - |\mathcal{N}'_u|\le 2k$.
The probability that $(u,z_i)$ is a fixed edge $(u,z)$ is
$\Pr{(u,z_i)=(u,z)}=\Pr{z_i=z}=\frac{1}{|\mathcal{N}'_u|}\le\frac{1}{|\mathcal{N}_u|-2k}=\frac{1}{d_u-2k}$.
Using a union bound and since $d_u>x=8ck\log(n/\delta)$ we obtain

\begin{align*}
    \Pr{z_i \in V_M }
  &   \le \sum_{y\in \mathcal{N}'_u \cap V_M} \Pr{z_i=y}
    \le \sum_{y\in \mathcal{N}'_u \cap V_M} \frac{1}{d_u-2k}
\\
&
    \le \frac{2k}{d_u-2k}\le \frac{1}{2c\log(n/\delta)}
    \le \frac{1}{2c}\enspace .
\end{align*}

Therefore the probability that $z_i$ is an exposed vertex, i.e., $z_i\notin V_M$
is $\Pr{z_i \notin V_M }\ge 1- \frac{1}{2c}$.

We define an indicator variable $I_i$ for queried edge $(u,z_i)$ for $i\in[y]$
which is one if $z_i \notin V_M $ and zero otherwise.
Note that $\Pr{I_i=1}\ge 1- \frac{1}{2c}$. Let $I=\sum_{i=1}^y I_i$.
Then, since $y$ $\ell_0$-samplers of $S_u$ use independent hash functions
we obtain

\[
  \begin{split}
   \Pr{I=0}
   &= \Pr{z_1\in V_M \wedge \cdots \wedge z_i\in V_M \wedge \cdots \wedge z_y\in V_M}\\
   &= \prod_{i=1}^{y} \Pr{z_i \in V_M}
   \le (\frac{1}{2c})^y= (\frac{1}{2c})^{8c\log(n/\delta)}
   \le \frac{\delta}{2n^c} \enspace .
  \end{split}
\]

Therefore, the probability that there exists an exposed vertex $z\in \{z_1,\cdots,z_y\}$
is $1-\frac{\delta}{2n^c}$.
Overall, the probability that sketch $S_u$
does not return FAIL and there exists an exposed vertex
$z\in \{z_1,\cdots,z_y\}$ is

\smallskip
{$\displaystyle
  \Pr{NoFAIL \wedge \{z_1,\cdots,z_y\}\backslash V_M\neq \emptyset}\ge 1-\delta/n^c\enspace.$}
\end{proof}

\begin{lemma}
\label{lem:timestamps:case:6}
Suppose that we invoke {\sf Rematch$((u,v),t)$}, and
 before time $t$, Invariants $1$, $2$ and $3$ hold,
and matching $M$ is maximal.
Then after time $t$, Invariants $1$, $2$ and $3$ hold and
matching $M$ is maximal.
The running time of  {\sf Rematch$((u,v),t)$} is $O(k\log^4(n/\delta))$.
\end{lemma}

\begin{proof}
First of all, we invoke {\sf AddEdgeToMatching$((u,v),t)$} to add edge $(u,v)$
to matching $M$. In Procedure {\sf AddEdgeToMatching$((u,v),t')$}, we insert the edge to
$M$ as well as $\mathcal{T}$ for some $t'\le t$.
We also insert $(u,v)$ to the sketch of
whichever vertex ($u$ or $v$) was exposed before time $t'$.
So at the end of {\sf AddEdgeToMatching$((u,v),t')$} edge $(u,v)$ is in $S_u$, $S_v$ and $\mathcal{T}$.

Once we invoke, Procedure {\sf DeleteFromDS$((u, v))$}, it deletes edge $(u,v)$ from
$S_u$, $S_v$ and $\mathcal{T}$. We also delete the edge from $M$. So after invocation of
{\sf DeleteFromDS$((u, v))$}, Invariants $1$, $2$ and $3$ hold.
Let us fix vertex $u$. The following proof is the same for vertex $v$.
We consider two cases for $u$.

(i) First, $u$ is a low-degree vertex, i.e., $d_u\le x$
assuming Invariants $1$, $2$ and $3$ hold. Observe that using
Lemma \ref{lem:low:rematch}, after the invocation of Procedure
{\sf Rematch$((u,v), t)$}, matching $M$ is maximal.
Moreover, the running time of {\sf Rematch$((u,v), t)$} when $u$ is
a low-degree vertex is $O(x\log^2 n\log(n/\delta))=O(x\log^3(n/\delta))$.
Let $\mathcal{N}'_u$ be the set of neighbors of $u$ that up to time $t$
are inserted into sketch $S_u$ but not deleted from $S_u$.
By Definition~\ref{def:ksparse},
by querying $S_u$ and with probability at least $1-\frac{\delta}{2n^c}$, we can recover $\mathcal{N}'_u$.
Observe that assuming Invariants $1$, $2$ and $3$ hold, we must have
$(\mathcal{N}_u \setminus \mathcal{N}'_u) \subseteq V_M$.
That is, those neighbors of
$u$ that are not in $\mathcal{N}'_u$ at time $t$ must be matched.
Therefore, all exposed neighbors of $u$ must be in $\mathcal{N}'_u$.
We have two sub-cases. First, if there is an exposed
$z \in\mathcal{N}'_w$ then we invoke {\sf AddEdgeToMatching$((w,z), t)$}.
Lemma \ref{lem:timestamps:case:1} shows that Invariants $1$, $2$ and $3$ hold
after invocation of {\sf AddEdgeToMatching$((w,z), t)$}.
The second subcase is if there is no exposed node in $\mathcal{N}'_w$, we then invoke
{\sf DeleteNeighborhood$(\text{vertex } w)$}.
Lemma \ref{lem:timestamps:case:5} shows that Invariants $1$, $2$ and $3$ hold
after invocation of {\sf DeleteNeighborhood$(\text{vertex } w)$}.

(ii) Second, $u$ is a high-degree vertex assuming Invariants $1$, $2$ and $3$ hold.
Observe that using Lemma \ref{lem:high:rematch}, after the invocation of Procedure
{\sf Rematch$((u,v), t)$}, matching $M$ with probability at least $1-\delta/n^c$
is maximal and the running time of Procedure {\sf Rematch$((u,v), t)$}
when $u$ is a high-degree vertex is $O(\log^4(n/\delta))$.
Since with probability at least $1-\delta/n^c$ there exists an exposed vertex
$z\in \{z_1,\cdots,z_y\}$, with this probability we invoke {\sf AddEdgeToMatching$((w,z), t)$}.
Lemma \ref{lem:timestamps:case:1} then shows that Invariants $1$, $2$ and $3$ hold
after invocation of {\sf AddEdgeToMatching$((w,z), t)$}.
\end{proof}


\subsection{Completing the Proof of Theorem \ref{thm:vc:max:k}}

First we prove the claim for the space complexity of our algorithm.
We maintain at most $2k$ sketches (for matched vertices), each one is
an $x$-sample recovery sketch for $x=8ck\log(n/\delta)$.
From Definition~\ref{def:ksparse} and the proof of Lemma
\ref{lem:timestamps:case:6}, the space to maintain an
$x$-sample recovery sketch is $O(k\log^4(n/\delta))$. So, we need
$O(k^2\log^4(n/\delta))$ bits of space to maintain the sketches of matched vertices.
The size of data structure $\mathcal{T}$, i.e.,
the number of edges stored in $\mathcal{T}$ is $|\mathcal{T}|\le (2k)^2$.
Thus, overall the space complexity of our algorithm is $O(k^2\log^4(n/\delta))$ bits.

Next we prove the update time and query time of our dynamic algorithm
for maximal matching is $\tilde{O}(k)$. In fact, the deletion or the insertion
time of an edge $(u,v)$ is dominated by the running time of most expensive
procedures which are {\sf AnnounceNeighborhood$(u)$}, {\sf DeleteNeighborhood$(u)$},
and {\sf Rematch$((u,v),t)$}. The running times of these procedures are also
dominated by the time to query at most $x$ edges from sketches $S_u$ and $S_v$ plus
the time to search for $x$ edges in data structure $\mathcal{T}$.

The time to query at most $x$ edges from sketches $S_u$ and $S_v$ using
Lemma \ref{lem:timestamps:case:6} is $O(k\log^4(n/\delta))$. The time to search
for $x$ edges in data structure $\mathcal{T}$ is $O(x\log k)=O(k\log^2(n/\delta))$ as we assume
the insertion, deletion and query times of $\mathcal{T}$ are all worst-case $O(\log k)$.
Therefore, the update time and query time of our dynamic algorithm
for maximal matching is $O(k\log^4(n/\delta))$.

Finally, we give the correctness proof of Theorem \ref{thm:vc:max:k}.
Observe that since after every time $t$ of stream $S$, Invariants $1,2$ and $3$ hold, and hence the matching $M$
is maximal. In fact, since Invariant $1$ holds,  for every edge $(u,v)\in E_t$
we have at least one of $v\in \mathcal{N}'_u$ or $u\in \mathcal{N}'_v$
which means $M$ is maximal.
Recall that $V_M$ is the set of vertices of matched edges in $M$.
Note that for every matched vertex $u\in V_M$,
we maintain an $x$-sample recovery sketch $S_u$.

Next, similar to the algorithm of Theorem~\ref{THM:VC:INSERTION} (Section \ref{sec:psa:vc}) we construct a graph $(G_M=(V_M,E_M),k)$.
For every matched vertex $v$, we extract up to $k$ edges incident on $v$ from
sketch $S_v$ and store them in set $E_M$. At the end,
we run the kernelization algorithm of Section~\ref{sec:kernel:VC}
on instance $(G_M=(V_M,E_M),k)$. The rest of proof of correctness of
Theorem~\ref{thm:vc:max:k} requires showing that maintaining a maximal
matching is sufficient to obtain a kernel for vertex cover, which is what was exactly argued in proof of Theorem~\ref{THM:VC:INSERTION}.



\section{Dynamic Parameterized Streaming Algorithm (DPSA for $VC(k)$ }
\label{app:vc:dpsa}

In this section we prove Theorem~\ref{THM:VC:DPSA}, which is restated below:

\begin{reptheorem}{THM:VC:DPSA}
Let $S$ be a dynamic parameterized stream of  insertions and deletions of edges of an underlying graph $G$.
There exists a randomized $(nk, nk+2^{2k^2})$-DPSA for $VC(k)$ problem.
\end{reptheorem}
\begin{proof}
Let $S$ be a stream of insertions and deletions of edges to an
underlying graph $G(V,E)$.
We maintain a
$kn$-sample recovery algorithm (Definition~\ref{def:ksparse}), which
processes all the edges seen in the stream; we also keep a counter
to record the degree of the vertex.
At the end of the stream $S$, we recover a graph $G'$ by extracting
the at most $kn$ edges from the recovery algorithm data structure, or
outputting ``NO'' if there are more than $kn$ edges currently in the
graph.
We then run the kernelization algorithm of Section~\ref{sec:kernel:VC}
on instance $(G', k)$.

Observe that if a graph has a vertex cover of
size at most $k$, then there can be at most $nk$ edges.
Each node in the cover has degree at most $n$, and every node must
either be in the cover, or be adjacent to a node in the cover.
Therefore, if the graph has more than $nk$ edges, it cannot have a
vertex cover of size $k$.
We take advantage of this fact to bound the overall cost of the
algorithm in the dynamic case.
We maintain a structure which allows us to recover at most $nk$ edges
from the input graph, along with a counter for the current number of
``live'' edges.
This can be implemented using a $k$-sample recovery algorithm
(Definition~\ref{def:ksparse}), or indeed by a deterministic algorithm
(e.g. Reed-Solomon syndromes).

The algorithm now proceeds follows.
To test for a vertex cover of size $k$, we first test whether the
number of edges is above $nk$: if so, there can be no such cover, and
we can immediately reject.
Otherwise, we can recover the full graph, and pass the graph to the
standard kernelization algorithm (Section~\ref{sec:kernel:VC}).
The total time for this algorithm is then $O(nk + 2^{2k^2})$, and the
space used is that to store the $k$-sample recovery algorithm, which
is $\tilde{O}(nk)$.

This assumes that each edge is inserted at most once, i.e. the same
edge is not inserted multiple times without intervening deletion.
This assumption can be removed, if we replace the edge counter with a
data structure which counts the (approximate) number of distinct edges
currently in the data structure.
This can provide a constant factor approximation with polylogarithmic
space.
This is sufficient to determine if the number of edges is greater than
$nk$, and if not, to recover the at most (say) $1.01nk$ edges in
the graph from the data structure storing the edges, and apply the kernelization algorithm of Section~\ref{sec:kernel:VC}.
\end{proof}

\junk{
\paragraph{[Alternate] Complexity of the algorithm.}
The space complexity of the algorithm is $\tilde{O}(nk)$: for each
node, we keep a $k$-sample recovery algorithm, which requires space
$\tilde{O}(k)$.
Tracking the degree of each vertex consumes space $O(n)$ in total.
Updating the $k$-sample algorithm takes time $\tilde{O}(1)$ per
update, be it insertion or deletion of an edge.
Querying each $k$-sample algorithm takes time linear in its size,
$\tilde{O}(nk)$ overall.
Thus, the cost of the query operation is dominated by the time to run
the kernel algorithm, which in turn is dominated by the brute-force
search step, which takes $O(2^{2k^2})$.

\paragraph{[Alternate] Correctness.}
First, we argue that each $k$-sample algorithm successfully recovers
(up to) $k$ true edges incident on each node, i.e., no instance fails
to recover, reports fewer than $\min(k,d_v)$ true nodes adjacent to $v$,
and none reports false positive neighbors.
For a single instance, the probability of such failures can be made
polynomially small within the $\tilde{O}(k)$ space bounds, so summing
these up, we obtain small failure probability overall via the union
bound.
Thus we assume that these algorithms all perform correctly, except
with this small probability.

Next, we argue that on the recovered graph $G'$, the kernelization
algorithm performs correctly: either it finds a vertex cover for $G$ of size
at most $k$, or correctly reports that no such cover exists.
We proceed by arguing that the kernelization algorithm described in
Section~\ref{sec:kernel:VC} has access to the same information from
$G'$ as it would for $G$, and therefore can follow a corresponding
execution path.
Consequently, the outcome must be the same.

We consider each of the kernelization rules in turn.
If (1) applies, there is a vertex $v \in G'$ with $d_v > k'$.
This must be a vertex whose degree in $G$ was also more than $k'$.
Accordingly, we remove $v$ and all edge incident on $v$ from $G'$, and
decrease $k'$ by one.
We call this a ``node removal''.
Note that there may well be edges present in $G$ incident on $v$ which
are not present in $G'$.
However, their absence in $G'$ does not affect the execution of the
algorithm, and they are effectively removed following the selection of
$v$ to the vertex cover.

If there is an isolated vertex $v \in G'$ (and hence is removed by rule (2)), then it would also be isolated
in $G$ following the sequence of applications of rule (1).
This is since each application of rule (1) decreases the degree of any
other vertex by at most 1, and we apply rule (1) at most $k$ times.
The only reason $v$ could become isolated in $G'$ but not in $G$ after
following the same set of node removals under rule (1) on both graphs
is if $v$ had degree $d'_v$ in $G'$ but degree $d_v$ in $G$, and was
subject to more than $d'_v$ removals of adjacent nodes.
However, observe that if $d_v \le k$, then $d'_v = d_v$, i.e. the
degree of $v$ in $G'$ is preserved, since all incident edges are
recovered correctly by the $k$-sample algorithm.
Hence, we only have $d'_v < d_v$ if $d_v > k$, and so $d'_v$ cannot
reach 0, since there cannot be this many node removals.

If we can apply rule (3), then the modified $G'$ has more than $k'^2$
edges.  Since the edges of $G'$ are a subset of the edges of $G$, then
necessarily, there can be no vertex cover of $G$ of size at most
$k'$.
This then leaves us with the case that no rules apply, and we have
reached a kernel of $G'$  with at most $2k$ vertices and at most $k'^2 \leq k^2$ edges.
We now observe that we must have the full connectivity information on
this set of at most $2k$ vertices.
Let $W$ be the set of removed nodes, and let $G \setminus W$ be the
graph $G$ with $W$ (and all edges incident on nodes in $W$) removed.

Suppose that some edge $(u,v)$ is present in $G \setminus W$ but
absent in $G' \setminus W$.
This can only occur if the degree of both $u$ and $v$ in $G$ was high:
if $d_u \le k$, then the edge $(u,v)$ would have been captured by $u$ in its
$k$-sample algorithm; and similarly for $v$.
However, if it was the case that $d_u > k$, then we could apply rule
(1) on $u$.
Note that if initially $d_u > k$, then it is the case that after each
application of (1), it remains that $d_u > k'$, since the degree of
$u$ in the modified graph decreases by at most 1, and $k'$ decreases
by 1 in each application of (1).
So if $u$ survives (is not covered by another node), then it can still
be picked by rule (1), contradicting the assumption that no rule can
be applied.
Therefore, if we reach the state that no rule can be applied, then we
must have the full adjacency information of the remaining kernel,
i.e. $G \setminus W = G'\setminus W$, and so we reach the same
conclusion when applying the exhaustive search.
Thus we correctly determine the solution to the $VC(k)$ instance from
$G'$.
}

\section{Concluding Remarks}

By combining techniques of kernelization with randomized sketch structures, we have initiated the study of \emph{parameterized streaming algorithms}. We considered the widely-studied Vertex Cover problem, and obtained results in three models: insertion only stream, dynamic stream and promised dynamic stream.
There are several natural directions for further study. We mention some of the below.

\paragraph{Dynamic Algorithms.} Recent work has uncovered connections between streaming algorithms and
dynamic algorithms~\cite{KKM13}.
It is natural to ask whether we can make the algorithms provided
dynamic: that is, ensure that after each step they provide (implicitly
or explicitly) a current answer to the desired problem.
The current algorithm for maximal matching sometimes takes time
polynomial in $k$ to process an update: can this be made sublinear in
$k$?

Our main algorithm in Section~\ref{sec:promised:dynamic} applies in the case where
there is a promise on the size of the maximal matching.
Can this requirement be relaxed?
That is, is there a dynamic algorithm that will succeed in finding a
maximal matching of size $k$ at time $t$, even if some intermediate
maximal matching has exceeded this bound?  Or can the cost be made
proportional to the largest maximal matching encountered, i.e. remove
the requirement for $k$ to be specified at the time, and allow the
algorithm to adapt to the input instance.

\paragraph{Other Problems.} In this paper, we only considered the
Vertex Cover problem. We think it is interesting to consider other
NP-hard problems in the framework of parameterized streaming, and that
kernelization algorithms can be helpful in this endeavour. In some
cases, one might be able to obtain parameterized streaming algorithms
with simple observations. For example, in the Feedback Vertex Set
($FVS(k)$) problem, we are given a graph $G=(V, E)$ and an integer
$k$. The question is whether there exists a set $V'\subseteq V$ such
that $G\setminus V'$ has no cycles.
We can show the following results (proved in the appendix)
for $FVS(k)$:

\begin{theorem}
\label{thm:fvs:upper}
There is a deterministic PSA for $FVS(k)$ which uses $O(nk)$ space.
\end{theorem}

\begin{theorem}
\label{thm:fvs:lower}
Any (randomized) PSA for $FVS(k)$ requires $\Omega(n)$ space.
\end{theorem}


\paragraph{Acknowledgments.}
The third author would like to thank Marek Cygan for fruitful discussion
on early  stages of this project in a Dagstuhl workshop.
We thank Catalin Tiseanu for some useful discussions regarding the
Feedback Vertex Set problem.


\newcommand{\Proc}{Proceedings of the~}

\newcommand{\STOC}{Annual ACM Symposium on Theory of Computing (STOC)}
\newcommand{\FOCS}{IEEE Symposium on Foundations of Computer Science (FOCS)}
\newcommand{\SODA}{Annual ACM-SIAM Symposium on Discrete Algorithms (SODA)}
\newcommand{\SOCG}{Annual Symposium on Computational Geometry (SoCG)}
\newcommand{\ICALP}{Annual International Colloquium on Automata, Languages and Programming (ICALP)}
\newcommand{\ESA}{Annual European Symposium on Algorithms (ESA)}
\newcommand{\CCC}{Annual IEEE Conference on Computational Complexity (CCC)}
\newcommand{\RANDOM}{International Workshop on Randomization and Approximation Techniques in Computer Science (RANDOM)}
\newcommand{\PODS}{ACM SIGMOD Symposium on Principles of Database Systems (PODS)}
\newcommand{\SSDBM}{ International Conference on Scientific and Statistical Database Management (SSDBM)}
\newcommand{\ALENEX}{Workshop on Algorithm Engineering and Experiments (ALENEX)}
\newcommand{\BEATCS}{Bulletin of the European Association for Theoretical Computer Science (BEATCS)}
\newcommand{\CCCG}{Canadian Conference on Computational Geometry (CCCG)}
\newcommand{\CIAC}{Italian Conference on Algorithms and Complexity (CIAC)}
\newcommand{\COCOON}{Annual International Computing Combinatorics Conference (COCOON)}
\newcommand{\COLT}{Annual Conference on Learning Theory (COLT)}
\newcommand{\COMPGEOM}{Annual ACM Symposium on Computational Geometry}
\newcommand{\DCGEOM}{Discrete \& Computational Geometry}
\newcommand{\DISC}{International Symposium on Distributed Computing (DISC)}
\newcommand{\ECCC}{Electronic Colloquium on Computational Complexity (ECCC)}
\newcommand{\FSTTCS}{Foundations of Software Technology and Theoretical Computer Science (FSTTCS)}
\newcommand{\ICCCN}{IEEE International Conference on Computer Communications and Networks (ICCCN)}
\newcommand{\ICDCS}{International Conference on Distributed Computing Systems (ICDCS)}
\newcommand{\VLDB}{ International Conference on Very Large Data Bases (VLDB)}
\newcommand{\IJCGA}{International Journal of Computational Geometry and Applications}
\newcommand{\INFOCOM}{IEEE INFOCOM}
\newcommand{\IPCO}{International Integer Programming and Combinatorial Optimization Conference (IPCO)}
\newcommand{\ISAAC}{International Symposium on Algorithms and Computation (ISAAC)}
\newcommand{\ISTCS}{Israel Symposium on Theory of Computing and Systems (ISTCS)}
\newcommand{\JACM}{Journal of the ACM}
\newcommand{\LNCS}{Lecture Notes in Computer Science}
\newcommand{\RSA}{Random Structures and Algorithms}
\newcommand{\SPAA}{Annual ACM Symposium on Parallel Algorithms and Architectures (SPAA)}
\newcommand{\STACS}{Annual Symposium on Theoretical Aspects of Computer Science (STACS)}
\newcommand{\SWAT}{Scandinavian Workshop on Algorithm Theory (SWAT)}
\newcommand{\TALG}{ACM Transactions on Algorithms}
\newcommand{\UAI}{Conference on Uncertainty in Artificial Intelligence (UAI)}
\newcommand{\WADS}{Workshop on Algorithms and Data Structures (WADS)}
\newcommand{\SICOMP}{SIAM Journal on Computing}
\newcommand{\JCSS}{Journal of Computer and System Sciences}
\newcommand{\JASIS}{Journal of the American society for information science}
\newcommand{\PMS}{ Philosophical Magazine Series}
\newcommand{\ML}{Machine Learning}
\newcommand{\DCG}{Discrete and Computational Geometry}
\newcommand{\TODS}{ACM Transactions on Database Systems (TODS)}
\newcommand{\PHREV}{Physical Review E}
\newcommand{\NATS}{National Academy of Sciences}
\newcommand{\MPHy}{Reviews of Modern Physics}
\newcommand{\NRG}{Nature Reviews : Genetics}
\newcommand{\BullAMS}{Bulletin (New Series) of the American Mathematical Society}
\newcommand{\AMSM}{The American Mathematical Monthly}
\newcommand{\JAM}{SIAM Journal on Applied Mathematics}
\newcommand{\JDM}{SIAM Journal of  Discrete Math}
\newcommand{\JASM}{Journal of the American Statistical Association}
\newcommand{\AMS}{Annals of Mathematical Statistics}
\newcommand{\JALG}{Journal of Algorithms}
\newcommand{\TIT}{IEEE Transactions on Information Theory}
\newcommand{\CM}{Contemporary Mathematics}
\newcommand{\JC}{Journal of Complexity}
\newcommand{\TSE}{IEEE Transactions on Software Engineering}
\newcommand{\TNDE}{IEEE Transactions on Knowledge and Data Engineering}
\newcommand{\JIC}{Journal Information and Computation}
\newcommand{\ToC}{Theory of Computing}
\newcommand{\MST}{Mathematical Systems Theory}
\newcommand{\Com}{Combinatorica}
\newcommand{\NC}{Neural Computation}
\newcommand{\TAP}{The Annals of Probability}
\newcommand{\TCS}{Theoretical Computer Science}

\bibliographystyle{plain}
\bibliography{gem}

\appendix

\section{Feedback Vertex Set}
\label{app:FVS}
In this section, we prove Theorem~\ref{thm:fvs:upper} and Theorem~\ref{thm:fvs:lower}



\subsection{Parameterized Streaming Algorithm (PSA) for $FVS(k)$}

We restate and prove Theorem~\ref{thm:fvs:upper} below:

\begin{reptheorem}{thm:fvs:upper}
There is a \emph{deterministic} PSA for $FVS(k)$ which uses $O(nk)$ space.
\end{reptheorem}


\begin{proof}
To prove Theorem~\ref{thm:fvs:upper}, we use the following lemma bounds the number of edges of a graph with small feedback vertex set.

\begin{lemma}
Any graph with a feedback vertex set of size at most $k$ can have at most $n(k+1)$ edges, where $n$ is the number of vertices of the graph.
\end{lemma}
\begin{proof}
Let the graph be $V=(G,E)$ and $S\subseteq V$ be the feedback vertex
set of size at most $k$. Then the graph $G\setminus S$ is a forest,
and hence has at most $n-|S|-1$ edges. Now each of the vertices in $S$
is adjacent to at most $n-1$ vertices in $G$. Hence the total number
of edges of $G$ is at most $(n-|S|-1)+(n-1)|S| = n+(n-2)|S|-1 \leq n+nk$ since $|S|\leq k$.
\end{proof}

\noindent The PSA algorithm for $FVS(k)$ runs as follows:
\begin{itemize}
\item Store all the edges that appear in the stream.
\item If the number of edges exceeds $n(k+1)$, output NO.
\item Otherwise the total number of edges (and hence the space complexity) is $n+nk$. Now that we have stored the entire graph, use any one of the various known FPT algorithms~\cite{DBLP:conf/swat/CaoCL10, DBLP:journals/corr/KociumakaP13} to solve the $FVS(k)$ problem.
\end{itemize}
This concludes the proof of Theorem~\ref{thm:fvs:upper}.
\end{proof}

\subsection{$\Omega(n)$ Lower Bound for $FVS(k)$}

We restate and prove Theorem~\ref{thm:fvs:lower} below:

\begin{reptheorem}{thm:fvs:lower}
Any (randomized) PSA for $FVS(k)$ requires $\Omega(n)$ space.
\end{reptheorem}

\begin{proof}
We show the proof by reduction to the {\textsc Disjointness} problem
in communication complexity.

\begin{center}
\noindent\framebox{\begin{minipage}{0.95\columnwidth}
\textsc{Disjointness}\\
\emph{Input}: Alice has a string $x\in \{0,1\}^n$ given by
$x_{1}x_{2}\ldots x_n$.\\ Bob has a string $y \in \{0, 1\}^n$.\\
\emph{Question}: Bob wants to check if $\exists i: x_{i}=y_i=1$.
\end{minipage}}
\end{center}

There is a lower bound of $\Omega(n)$ bits of communication between
Alice and Bob, even allowing randomization~\cite{nisan}.

\begin{figure}[t]
\centering
\includegraphics[width=\textwidth]{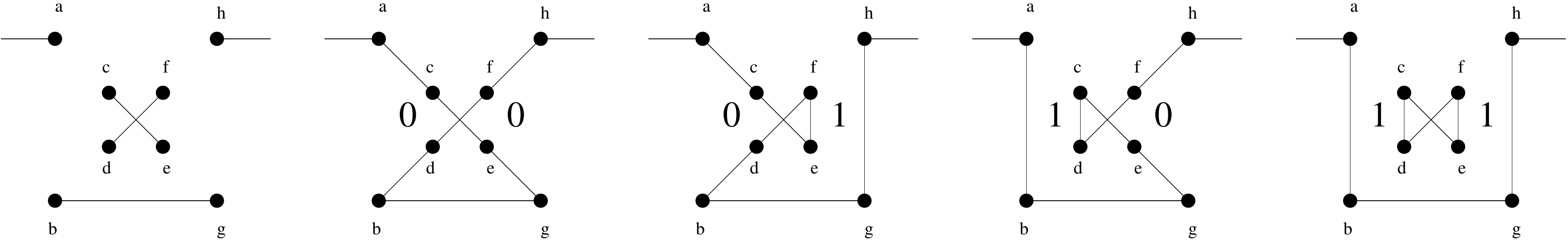}
\caption{Gadget for reduction to {\textsc Disjointness}}
\label{fig:disj}
\end{figure}

Given an instance of {\textsc Disjointness}, we create a graph on $8n$
nodes as follows.
We create nodes $a_i, b_i, \ldots h_i$, and insert edges
$(b_i, g_i), (c_i, e_i), (d_i, f_i)$ for all $i$.
We also create edges $(h_i, a_{i+1})$ for $1 \leq i < n$.
This is illustrated in the first graph in Figure~\ref{fig:disj}.

For each $i$, we add 2 edges corresponding to $x_i$, and two according
to $y_i$.
If $x_i = 0$, we add $(a_i, c_i)$ and $(b_i, d_i)$; else we add $(a_i,
b_i)$ and $(c_i, d_i)$.
If $y_i=0$, we add $(f_i, h_i)$ and $(e_i, g_i)$; else we add $(f_i,
e_i)$ and $(g_i, h_i)$.

We now observe that the resulting graph is a tree (in fact it is a
path) if the two strings are disjoint, but it has at least one cycle
if there is any $i$ such that $x_i =  y_i = 1$.
This can be seen by inspecting Figure~\ref{fig:disj}, which shows the
configuration for each possibility for $x_i$ and $y_i$.
Thus, any streaming algorithm that can determine whether a graph
stream is cycle-free or has one (or more) cycles implies a
communication protocol for {\sc Disjointness}, and hence requires
$\Omega(n)$ space.

Since $FVS(k)$ must, in the extreme case $k=0$, determine whether $G$
is acyclic, then $\Omega(n)$ space is required for this problem also.
This generalizes to any constant $k$
by simply adding $k$ triangles
on $3k$ new nodes to the graph: one node from each must be removed,
leaving the question whether the original graph is acyclic.
\junk{
First we give a lower bound for a relative of the {\textsc Index}
problem that we refer to as the \textsc{Index Same} problem.
\begin{center}
\noindent\framebox{\begin{minipage}{0.95\columnwidth}
\textsc{Index Same}\\
\emph{Input}: Alice has a string $x\in \{0,1\}^n$ given by $x_{1}x_{2}\ldots x_n$.\\ Bob has an index $i\in [n-1]$\\
\emph{Question}: Bob wants to check if $x_{i-1}=x_i$.
\end{minipage}}
\end{center}

\begin{lemma}
Any (randomized) one-way communication protocol between Alice and Bob for the \textsc{Index Same} problem requires $\Omega(n)$ bits.\label{thm:lower-index-same}
\end{lemma}
\begin{proof}
We give a reduction from the \textsc{Index} problem. Consider an instance $(x,j)$ of the \textsc{Index} problem, where $x\in \{0,1\}^n$ and $j\in [n]$. We now build a string $y\in \{0,1\}^{2n}$ by setting $y_{2i-1} = 0$ and $y_{2i}=1$ if $x_i = 0$ and $y_{2i-1}= 1 = y_{2i}$ otherwise. By the construction of the string $y$, it follows that $y_{2j-1}=y_{2j}$ if and only if $x_j=1$. Hence we can use the \textsc{Index Same} problem to solve the \textsc{Index} problem, and hence the lower bound carries over to the \textsc{Index Same} problem as well.
\end{proof}

The next lemma shows that if there is a single pass streaming
algorithm, say $\mathcal{B}$, which solves the problem of checking
whether a given forest on $n$ vertices is a tree
in $f(n)$ space, then there is a protocol for the \textsc{Index Same} instance with $f(n)$ bits. Hence the lower bound of $\Omega(n)$ transfers to the problem of checking whether a given forest is a tree.

\begin{lemma}
Consider a graph $G$ with $n-1$ edges on $n$ nodes. Then any (randomized) algorithm to determine whether $G$ is a tree requires $\Omega(n)$ space. \label{thm:lower-bound-checking-tree}
\end{lemma}
\begin{proof}
Consider an instance $(x,j)$ of the \textsc{Index Same} problem, where $x\in \{0,1\}^n$ and $j\in [n-1]$. We now build a graph as follows:
\begin{itemize}
\item For each $j\in [n]$, add a vertex $v_j$
\item Add two special vertices ZERO and ONE.
\item For each $i\in [n]$, if $x_i = 0$ then Alice adds the edge $(v_{i},\text{ZERO})$; otherwise if $x_i = 1$ then Alice adds the edge $(v_{i},\text{ONE})$.
\end{itemize}
The graph now has $n$ edges and $n+2$ vertices. Alice sends this edge
set to Bob using $f(n)$ bits.
Bob, holding the desired index $i$, now adds the edge $(v_{i},v_{i-1})$
Now we have $n+1$ edges and $n+2$ vertices. It is easy to see that the
constructed graph is connected (and hence a tree) if and only if
$x_{j}\neq x_{j+1}$, i.e., Bob can check solve the \textsc{Index Same}
problem by checking whether the output of the algorithm $\mathcal{B}$
indicates if $G$ is a tree.
The total communication between Alice and Bob was $O(f(n))$ bits, and hence we can solve the \textsc{Matrix Index} problem in $f(n)$ bits. This implies $f(k)=\Omega(n)$.
\end{proof}

Finally, we are now ready to prove Theorem~\ref{thm:fvs:lower}.
Let $\mathcal{A}$ be any PSA for $FVS(k)$. If a connected graph $G$ is
given as an input to $\mathcal{A}$, then deciding if $G$ is acyclic
(in the extreme case when $k=0$),
i.e., a tree, needs $\Omega(n)$ space from Lemma~\ref{thm:lower-bound-checking-tree}. Hence, any PSA for $FVS(k)$ needs $\Omega(n)$ space.
}
\end{proof}


\end{document}